%% file: Wireless_SGLD.tex
\documentclass[12pt, draftclsnofoot, onecolumn]{IEEEtran}
\IEEEoverridecommandlockouts

\usepackage{fancyhdr}
\usepackage{epsfig}
\usepackage{threeparttable}
\usepackage{epsf,epsfig}
\usepackage{amsthm}
\usepackage[cmex10]{amsmath}
\usepackage{amssymb, amsfonts}
\usepackage[noadjust]{cite}
\usepackage{dsfont}
\usepackage{subfigure}
\usepackage{color}
\usepackage{soul}
\usepackage{amssymb}
 \usepackage{accents}
  \usepackage{algorithm}
   \usepackage[english]{babel}
    \usepackage{url}
    \usepackage{framed} 
        \usepackage{bm}

\allowdisplaybreaks[4]

\begin{document}
\include{header}
\title{\huge Wireless Federated Langevin Monte Carlo: Repurposing Channel Noise for  Bayesian Sampling and Privacy}


\author{ \IEEEauthorblockN{Dongzhu Liu and Osvaldo Simeone   \vspace{-6 mm}  \footnote{D. Liu is with the School of Computing Science, University of Glasgow, UK (email: dongzhu.liu@glasgow.ac.uk). She was formerly with the Department of Engineering of Kings College London, UK.   O. Simeone  is with Department of Engineering of Kings College London, UK (email: osvaldo.simeone@kcl.ac.uk). He has received funding from the European Research Council (ERC) under the European Unions Horizon 2020 Research and Innovation Programme (Grant Agreement No. 725731). }}
}
\maketitle

\begin{abstract}

Most works on federated learning (FL) focus on the most common frequentist formulation of learning whereby the goal is minimizing the global empirical loss. Frequentist learning,  however, is known to be problematic in the regime of limited data as it fails to quantify epistemic uncertainty in prediction.  Bayesian learning provides a principled solution to this problem by shifting the optimization domain to the space of distribution in the model parameters. {\color{black}This paper proposes a novel mechanism for the efficient implementation of Bayesian learning in wireless systems. Specifically, we focus on a standard gradient-based Markov Chain Monte
Carlo (MCMC) method, namely Langevin Monte Carlo (LMC), and we introduce a novel protocol, termed Wireless Federated LMC (WFLMC), that is able to repurpose channel noise for the double role of seed randomness for MCMC sampling and of privacy preservation.} To this end, based on the analysis of the Wasserstein distance between sample distribution and global posterior distribution under privacy and power constraints, we introduce a power allocation strategy as the solution of a convex program. The analysis identifies distinct  operating regimes in which the performance of the system is power-limited, privacy-limited, or limited by the requirement of MCMC sampling. Both analytical and simulation results  demonstrate that, if the channel noise is properly accounted for under suitable conditions, it can  be fully repurposed for both MCMC sampling and privacy preservation, obtaining the same performance as in an ideal communication setting that is not  subject to privacy constraints.

\end{abstract}


\section{Introduction}

Federated learning (FL) protocols aim at coordinating multiple devices to collaboratively train a target model in a manner that approximates centralized learning at the cloud, while avoiding the direct exchange of data \cite{park2019wireless, zhu2020toward}. 
Most prior works on wireless FL consider a \emph{frequentist} formulation whose goal is minimizing the empirical loss over the vector of model parameters \cite{sery2021over, yang2021revisiting,zhu2019broadband, liu2020privacy,zhu2020one}. Significant attention has been devoted to uncoded transmission schemes coupled with non-orthogonal multiple access (NOMA), which leverage the superposition property of wireless channels to enable efficient over-the-air aggregation at the server \cite{zhu2019broadband, liu2020privacy}.  Furthermore, FL protocols inevitably leak some information about local data via communication. Formal privacy requirements can be met by introducing randomness to the disclosed statistics \cite{dwork2014algorithmic}. When implementing uncoded transmission, noise in wireless channels was accordingly shown to serve as a privacy-preserving mechanism \cite{koda2020differentially,liu2020privacy}.

Frequentist learning is effective in the regime of large data sets when accuracy is the main concern, but it fails to quantify \emph{epistemic uncertainty} due to the availability of limited data \cite{guo2017calibration,lakshminarayanan2016simple}. \emph{Bayesian learning} provides an alternative learning framework in which optimization is done over the distribution of model parameters rather than over a single model parameter vector as in frequentist learning. 
Practical Bayesian learning methods include  \emph{variational inference} (VI), which constrains the model distribution to a parameter family,  and \emph{Monte Carlo} (MC) sampling,  which draws samples approximately generated from the optimal model distribution \cite{angelino2016patterns}. 

This paper represents the first work on Bayesian FL in wireless networks. {\color{black}We specifically adopt \emph{Langevin Monte Carlo} (LMC), a gradient-based \emph{Markov Chain Monte Carlo} (MCMC) method that adds Gaussian noise to \emph{gradient descent} (GD) updates. LMC is a fundamental building block of computationally efficient Bayesian inference and learning strategies. Unlike simpler random-walk MCMC methods, LMC leverages first-order information about the probabilistic model, striking a useful trade-off between complexity and performance \cite{angelino2016patterns, ma2015complete}. LMC can be generalized and improved in various directions, such as by accounting also for second-order information \cite{dalalyan2020sampling, zou2021convergence}. 

 The key contribution of this paper is not that of introducing a new Bayesian learning algorithm. Rather, we introduce a new mechanism for the efficient, and private, implementation of LMC over wireless channels. The approach is based on the idea that channel noise can be repurposed for the double role of seed randomness for the implementation of MC sampling and of privacy-preservation. Our analytical and experimental results provide insights about operating regimes in which channel noise can effectively serve both functions. It is envisaged that the proposed novel method of exploiting channel noise for MC sampling could also be applied and optimized for more sophisticated MCMC solutions such as Hamiltonian Monte Carlo \cite{zou2021convergence}.}
\subsection{Related Work}

\subsubsection{Frequentist and Bayesian FL} FL protocols alternate between local computing and communication steps. In frequentist FL protocols, devices exchange model parameter vectors, which may be first quantized and compressed \cite{alistarh2017qsgd}.  When implemented over wireless channels, FL can benefit from over-the-air aggregation via uncoded transmission -- an approach known as AirComp \cite{zhu2019broadband, sery2020analog,liu2020privacy}. AirComp can be combined with sparsification and compression to reduce the communication overhead \cite{fan2021temporal}.  Bayesian post-processing estimation methods have been proposed to improve the test accuracy of federated learning, e.g., by exploiting the temporal structure of the received signals \cite{fan2021temporal}, by incorporating information about channel distribution and local prior \cite{lee2020bayesian}, or by addressing data heterogeneity via knowledge distillation \cite{chen2020fedbe,zhu2021data}. Note that the schemes in  \cite{fan2021temporal,lee2020bayesian} do not implement Bayesian learning in the sense explained above of optimizing over a distribution in the model parameter space; while \cite{chen2020fedbe,zhu2021data} consider ideal communication and require either unlabelled data at the server \cite{chen2020fedbe} or additional communication~overhead~\cite{zhu2021data}.

As discussed, Bayesian learning is, in practice, implemented by approximate methods  -- either VI or MC sampling. Both have been investigated only to a very limited extent for FL, even in the presence of ideal communication. VI-based methods are proposed in  \cite{corinzia2019variational,kassab2020federated} for noiseless communications based on parametric and particle-based representations of the model parameter distribution.  Gradient-based MC methods are instead investigated in \cite{elfederated,vono2021qlsd}, again under ideal~communications. 


\subsubsection{Private FL}
Differential privacy (DP) is a strong measure of information leakage that relates to the sensitivity of the disclosed statistics on individual data points in the training data set. 
In FL, a standard model is to assume the edge server to be ``honest-but-curious", requiring the implementation of DP-preserving mechanisms such as noise addition, subsampling random mini-batches, and random quantization of the gradients \cite{gandikota2019vqsgd, agarwal2018cpsgd}.  
Wireless FL can repurpose channel noise so as to ensure DP guarantees by controlling the signal-to-noise (SNR) ratio via transmit power optimization \cite{koda2020differentially}.  Furthermore, the superposition property of NOMA not only achieves efficient aggregation, but also amplifies  the role of the channel noise as a privacy mechanism by protecting multiple devices' transmissions simultaneously \cite{seif2020wireless}.  To enhance the convergence rate under the DP constraints,  reference \cite{liu2020privacy} proposes an optimized adaptive power control strategy that increases the effective SNR over the iterations. We note that the presence of channel noise can also benefit learning by accelerating the convergence for non-convex models \cite{sery2021over,zhang2021turning}, or improving the generalization capability of convex models~\cite{yang2021revisiting}.  

\subsubsection{Private Bayesian learning} For Bayesian learning, the inherent randomness induced MC sampling automatically satisfies some level of DP requirements. Specifically, producing a single sample from the exact (or approximate) posterior distribution implements a differentially private strategy known as the exponential mechanism
 \cite{wang2015privacy}. This result can be extended to multiple samples in gradient-based MCMC under proper conditions, such as small learning rate \cite{wang2015privacy} or large scale model~\cite{li2019connecting}. 
 All prior work on private Bayesian learning is limited to centralized settings, and no prior result appears to have studied privacy in the context of Bayesian FL. 

\subsection{Contributions and Organization}
In this paper, we 
introduce a federated implementation of LMC in wireless systems whereby  power allocation is optimized to control the SNR level so as to meet the requirement of both MC sampling and DP. The main contributions and findings of the paper are summarized as follows.

\noindent $\bullet$ {\bf Introducing Wireless Federated Langevin Monte Carlo (WFLMC):} We first introduce Wireless Federated Langevin Monte Carlo (WFLMC), a novel iterative Bayesian learning protocol that relies on power control to repurpose channel noise for the double role of MC sampling via LMC and privacy preservation. WFLMC is based on uncoded transmission and NOMA, and goes beyond existing frequentist AirComp strategies by quantifying epistemic uncertainty through Bayesian learning.

\noindent $\bullet$ {\bf Analyzing WFLMC:}  Unlike frequentist learning, the goal of MC sampling in Bayesian learning is to ensure that the distribution of the produced samples is close to the global posterior distribution. Accordingly, we measure the learning performance via the 2-Wasserstein distance between the two distributions as in \cite{dalalyan2017further}. We provide analytical bounds on the 2-Wasserstein distance that comprise the contribution of the discretization error incurred by LMC, as well as of the gradient error due to channel noise and scheduling. We also present a DP privacy analysis of WFLMC that provides insights into the impact of the channel noise  on the privacy loss.

\noindent $\bullet$ {\bf Optimized power allocation and scheduling:} Building the analytical results, we formulate the optimization of the power allocation and scheduling policy as the minimization of  the 2-Wasserstein distance under DP and power constraints. The resulting optimization is shown to be a convex program,  and a closed-form solution is provided under simplifying assumptions. The analysis identifies distinct operating regimes in which the performance is power-limited, DP-limited, or LMC-limited. The three regimes are determined by the relative values of transmitted power, privacy level, and learning rate. The analytical results demonstrate that in the LMC-limited regime channel noise can be fully repurposed for both MC sampling and privacy preservation, obtaining the same performance as in an ideal communication setting that is not subject to DP constraints. 
For the general case, we formulate a min-max problem that can be converted into a convex problem. 

\noindent $\bullet$ {\bf Experiments:} We provide extensive numerical results to demonstrate the joint role of channel noise for MC sampling and privacy. 

In closing this section, we would like to emphasize the relationship of this work with our previous papers \cite{liu2020privacy} and  \cite{liu2021channel}. As mentioned, in \cite{liu2020privacy}, we considered frequentist FL on a wireless channel, and analyzed the problem of optimal power allocation for an AirComp-based strategy that leverages channel noise as a privacy mechanism. The problem formulation has a minor overlap with the setting studied here, which focuses on Bayesian learning. In fact, Bayesian learning requires the analysis of performance metrics based on distributions in the model parameter space \cite{dalalyan2017further}, and it cannot rely on the standard tools for the convergence of gradient-based schemes used in \cite{liu2020privacy}. In contrast, reference \cite{liu2021channel} introduces a one-shot Bayesian protocol for a wireless data center setting in which the server has access to the global data set. In this system, the global data set is divided up among the workers to benefit from computational parallelism, but the server uses its access to the global data set during training. Specifically, paper \cite{liu2021channel} proposes a novel VI-based strategy that builds on consensus MC \cite{rabinovich2015variational} by accounting for the presence of fading and channel noise. The contribution of \cite{liu2021channel} is distinct from the current manuscript for a number of reasons. First,  in the current work, we study for the first time iterative, rather than one-shot, Bayesian learning protocols. Second, we concentrate on a federated setting in which the server does not have access to the global data set. Third, no privacy constraints are assumed in \cite{liu2021channel}. And, fourth, unlike \cite{liu2021channel}, this work provides an analysis of the optimal power allocation strategy and draws theoretical conclusions on the capacity of the channel noise to serve the double role of seed randomness for MC sampling and privacy protection.

{\bf Organization:}  The remainder of the paper is organized as follows. 
Section~\ref{sec: system model} introduces the system model. 
Section~\ref{sec: WFLMC} proposes the design of WFLMC. 
Section~\ref{sec: analysis WFLMC} presents convergence and privacy analysis of WFLMC, while the optimal power allocation and scheduling are provided in Section \ref{sec: optmization}, followed by  numerical results in  Section~\ref{sec:simulation} and conclusions in Section~\ref{sec: conclusions}.

\section{System Model} \label{sec: system model}


\begin{figure}[t]
\centering
\includegraphics[width=15cm]{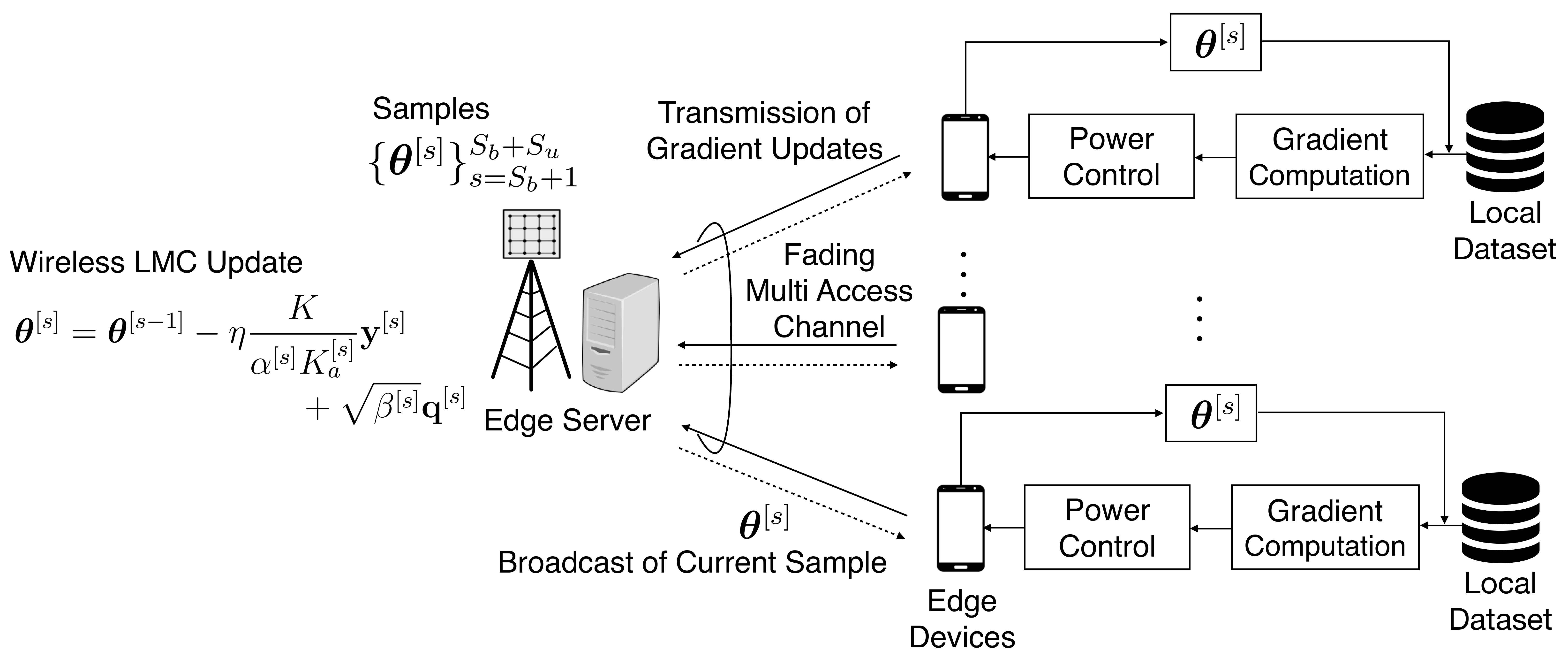}
\vspace{-6mm}
\caption{Differentially private federated Bayesian learning system based on Langevin Monte Carlo (LMC).}
\vspace{-6mm}
\label{Fig: FL system}
\end{figure}

As shown in Fig. \ref{Fig: FL system}, we consider a wireless federated edge learning system comprising a single-antenna edge server and $K$ edge devices connected through it via a shared wireless channel. Each device $k$, equipped with a single antenna, has its own local dataset $\cD_k$ encompassing $N_k$  data samples $\cD_k=\{ \bd_{k,n} \}_{n=1}^{N_k}$. For supervised learning applications, each data sample $\bd_n=(\bu_n, v_n)$ is in turn partitioned into a covariate vector $\bu_n$ and a label $v_n$; while, for unsupervised learning applications such as generative modeling, it consists of a single vector $\bd_n$. The global data set is denoted as  $\cD=\{\cD_k\}_{k=1}^K$. 
The goal of the system is to carry out Bayesian learning via gradient-based Monte Carlo (MC) sampling. Accordingly, through communication with devices, the server wishes to obtain a number of random samples of the model parameter vector ${\bm \theta}\in \mathbb{R}^m$ that are approximately distributed according to the global posterior distribution $p({\bm \theta} | \mathcal{D} )$. Unlike \cite{liu2021channel}, which studied one-shot protocols, the gradient-based MC methods studied in this paper are iterative, in a manner similar to standard federated learning protocols such as FedAvg. 
\vspace{-3mm}
\subsection{Langevin Monte Carlo}\vspace{-3mm}
The machine learning model adopted by the system is defined by a likelihood function  $p( \bd |{\bm \theta}  )$ as well as by a prior distribution $p(\bm \theta)$. Accordingly, the likelihood of the data at device $k$  is 
\begin{align}
 p(\cD_k|{\bm \theta}) = \prod_{n=1}^{N_k}  p(\bd_{n,k} |{\bm \theta}),
\end{align}
where the likelihoods $ p(\bd_{n,k} |{\bm \theta}) $ may be different across devices.
{\color{black}
The goal of Bayesian learning is to compute the global posterior 
\begin{align}\label{eq: posterior}
\text{(Global Posterior)} \quad p({\bm \theta}|\cD)  \propto p({\bm \theta}) \prod_{k=1}^K  p(\cD_k | {\bm \theta}).
\end{align}
The global posterior can be expressed in terms of the local posteriors at each device $k$, i.e.,
\begin{align}
\text{(Local Posterior)} \quad  p({\bm \theta}|\cD_k)  \propto p({\bm \theta})^{1/K}  p(\cD_k | {\bm \theta}) 
\end{align}
since we have the equality 
\begin{align}
p({\bm \theta}|\cD) \! \propto \! \prod_{k=1}^Kp({\bm \theta}|\cD_k). 
\end{align}
In contrast, frequentist learning is concerned with the optimization of the global cost function
\begin{align}\label{eq: gb func}
\text{(Global Cost Function)} \quad  f({\bm \theta })=-\log p({\bm \theta}|\cD)  = -\sum_{k=1}^K \log p({\bm \theta}|\cD_k), 
\end{align} 
which is the log-loss evaluated from the unnormalized posterior \eqref{eq: posterior}. Minimizing the global cost function \eqref{eq: gb func} yields the maximum a posterior (MAP) solution of frequentist learning. The global cost function \eqref{eq: gb func} can be expressed in terms of the local cost functions
\begin{align}\label{eq: local func}
\text{(Local Cost Function)} \quad f_k(\bm \theta)=-\log p(\cD_k|{\bm \theta}) -\frac{1}{K}\log p(\bm \theta),
\end{align} 
since we have the equality $ f(\bm \theta)= \sum_{k=1}^K f_k(\bm \theta)$.
} 

{\color{black}
Directly computing the global posterior distribution \eqref{eq: posterior} is generally of prohibitive complexity. To address this problem, Monte Carlo (MC) methods represent the global posterior distribution \eqref{eq: posterior} in terms of samples approximately distributed from it. Specifically, Markov Chain MC (MCMC) techniques produce a sequence of samples ${\bm \theta^{[s]}}$ with $s=1,2,...$ with the key property that as $s$ grows large, the marginal distribution of sample ${\bm \theta^{[s]}}$ tends to the desired posterior distribution. As discussed in Sec. I, in this paper, we specifically adopt Langevin MC (LMC), a fundamental MCMC technique that has been widely studied as a means to strike a practical balance between complexity and efficiency (see, e.g.,  \cite{angelino2016patterns, ma2015complete}).  LMC is a gradient-based MCMC sampling scheme building on the global cost function \eqref{eq: gb func}.  As we elaborate on next, the LMC update is derived as an approximation of a continuous-time differential process that has the desired  property of asymptotically producing samples drawn from the global posterior~\eqref{eq: posterior}. 

To start, we first introduce the continuous-time Langevin diffusion process  (LDP) $\{\bar{\bm \theta}^{(t)}:  t\in \mathbb{R}^+ \}$  follows the stochastic differential~equation
  \begin{align} \label{eq: LD}
\text{(LDP)} \quad   \mathrm{d} \bar{\bm \theta}^{(t)}  =   - \nabla f(\bar{\bm \theta}^{(t)} )  \mathrm{d} t +\sqrt{2} \mathrm{d} \boldsymbol{B}^{(t)} , 
 \end{align}
 where $\boldsymbol{B}^{(t)} $ represents  Brownian motion. The Langevin diffusion process \eqref{eq: LD} has the invariant stable distribution $p^*(\bm \theta) \!=\!p({\bm \theta}|\cD)\!\propto\! \exp[-f(\bm \theta)]$, which corresponds to the desired global posterior \eqref{eq: posterior} \cite{bhattacharya1978criteria}.}

 The integral of the stochastic differential equation \eqref{eq: LD} in the range of $t\in[s \eta, (s+1)\eta]$ for integers $s=1,2,\dots$ and duration $\eta$ yields
 \begin{align}\label{eq: diff LD}
\bar{\bm \theta}^{((s+1)\eta)} -  \bar{\bm \theta}^{(s\eta)} = -\int_{s\eta}^{(s+1)\eta}  \nabla f(\bar{\bm \theta}^{(t)} )  \mathrm{d} t + \sqrt{2\eta} {\bm \xi}^{[s+1]}, 
  \end{align}
  where $\{{\bm \xi}^{[s+1]}\}$ is a sequence of  identical and independent (i.i.d.) random vectors following the Gaussian distribution  $\cN(0, \bI_m)$.
 This is because we have the equality $ \sqrt{2}  \int_{s\eta}^{s\eta+\eta} \mathrm{d}\boldsymbol{B}^{(t)}= \sqrt{2} (\boldsymbol{B}^{(s\eta+\eta)}-\boldsymbol{B}^{(s\eta)})\sim\cN(0,2\eta)$. 
  
  In the limit of a sufficiently small $\eta$ such that the following approximation holds 
 \begin{align}\label{eq: int. discre}
 \int_{s\eta}^{(s+1)\eta}  \nabla f(\bar{\bm \theta}^{(t)} ) \mathrm{d} t \approx \eta  \nabla f( {\bm \theta}^{[s]} ) , 
 \end{align} 
  the discretization \eqref{eq: diff LD} results in Langevin MC (LMC),  a gradient-based Markov chain MC (MCMC) sampling scheme that proceeds according to the iterative update rule 
\begin{align} \label{eq: LMC}
\text{(LMC)} \quad {\bm \theta}^{[s+1]}={\bm \theta}^{[s]} -\eta \nabla f({\bm \theta}^{[s]})  + \sqrt{2\eta} {\bm \xi}^{[s+1]}, 
\end{align} 
where we have $   {\bm \theta}^{[s]} = \bar{\bm \theta}^{(s\eta)} $ and $\eta$ as step size.  The error due to approximation \eqref{eq: int. discre} was studied in \cite{dalalyan2017further, chatterji2018theory}, and it will be further discussed in Sec. \ref{sec: convergence analysis}. 
 
  In order to obtain samples ${\bm \theta}^{[s]}$ approximately drawn from the global posterior distribution, LMC discards the samples produced in the first $S_b$ iterations \eqref{eq: LMC}, also known as \emph{burn-in period}. The remaining $S_u$ samples ${\bm \theta}^{[s]}$ with $s=S_b+1, S_b+2, \dots,S_b+S_u,$ are retained and used for  downstream applications such as ensemble prediction.

 {\color{black}We consider LMC in this work due to its simplicity and scalability, as LMC updates amount to adding Gaussian noise to gradient descent updates.  Other gradient-based MCMC methods, such as kinetic MC  \cite{dalalyan2020sampling}  and Hamilton MC \cite{zou2021convergence}, which modify higher-order updates may also be studied in a manner similar to this paper and are left for future research.}

 \vspace{-3mm}
 \subsection{Learning Protocol}
 The goal of the system under study is to implement LMC \eqref{eq: LMC} in the described federated setting with $K$ devices, while satisfying formal differential privacy (DP) guarantees to be detailed in Sec. \ref{sec: metrics}.
The protocols is organized in iterations $s=1,2, \dots, S_b+S_u$ with $S_b$ denoting the burn-in period, across which the server maintains sample iterates ${\bm \theta}^{[s]}$.  {\color{black}The selection of burn-in period should avoid the regime in which the samples are too dependent on the Markov chain's initialization. A formal test for this purpose is the Gelman-Rubin diagnostic, which is based by using multiple  Markov chains. The diagnostic compares the estimated between-chains and within-chain variances for each model parameter, and chooses the burin-in period after which difference between theses variances are sufficiently small. }

At  each $s$-th communication round, the edge server broadcasts the current sample ${\bm \theta}^{[s]}$ to all edge devices via the downlink channel.  {\color{black}We assume that downlink communication is ideal, so that each device receives the sample ${\bm \theta}^{[s]}$ without distortion.  This assumption is practically well justified when the edge server communicates through a base station with less stringent power constraint than the devices and the use of the whole downlink bandwidth for broadcasting. It is commonly made in many related papers, such as \cite{zhu2019broadband,liu2020privacy,zhu2020one}. }  

By using the received vector ${\bm \theta}^{[s]}$ and the local dataset $\cD_k$, each device computes the gradient of the local cost function \eqref{eq: local func}  as
\begin{align}\label{eq: local gradient}
{(\text {Local gradient})} \quad   \nabla f_k\big({\bm \theta}^{[s]}\big)= -\sum_{n =1}^{N_k}\nabla \log p(\bd_n|{\bm \theta}^{[s]}) -\frac{1}{K} \nabla\log p(\bm \theta^{[s]}),
\end{align}
which is transmitted over the wireless shared channel to the edge server. The goal is to enable the edge server to approximate the update term in \eqref{eq: LMC}, namely 
\begin{align} \label{eq: LMC update term}
-\eta \nabla f({\bm \theta}^{[s]})  + \sqrt{2\eta} {\bm \xi}^{[s+1]}= -\eta \sum_{k=1}^K \nabla f_k({\bm \theta}^{[s]})  + \sqrt{2\eta} {\bm \xi}^{[s+1]} . 
\end{align}
As we will see, channel noise can be repurposed to contribute to the additive random term  ${\bm \xi}^{[s+1]}$ in the LMC update \eqref{eq: LMC}.    
The steps in \eqref{eq: local gradient} and \eqref{eq: LMC} are iterated across multiple communication rounds until a convergence condition is met. As a result, the server obtains a sequence of global model parameter vectors ${\bm \theta}^{[s]}$, with $s=1,2,\dots, S_b+S_u$. 


\subsection{Communication Model}\label{sec: comm model}

The devices communicate via the uplink to the edge server on the shared 
wireless channel. The proposed approach leverages analog transmission in order to: (i) benefit from over-the-air computing as in many prior works \cite{zhu2019broadband,sery2020analog,liu2020privacy}; and (ii) to repurpose channel noise for MC sampling and as a privacy mechanism.  
We assume a block flat-fading channel, where the channel coefficients remain constant within a communication block, and they vary in a potentially correlated way over successive blocks. 
 Each block contains $m$ channel uses, allowing the uncoded transmission of a gradient vector via  non-orthogonal multiple access (NOMA) as in  \cite{sery2020analog,liu2020privacy}. 

We assume symbol-level synchronization among the subset of devices that are scheduled in each block,  enabling over-the-air computing.  This can be achieved by using standard protocols such as the timing advance procedure in LTE and 5G NR \cite{mahmood2019time}. In the $s$-th communication round, the corresponding received signal is   
 \begin{align} \label{eq: rev signal}    
\by^{[s]}= \sum_{k=1}^K h_k^{[s]} \bx_k^{[s]} +\bz^{[s]},
\end{align}
where $h_k^{[s]}$ is the channel gain for device $k$ in round  $s$, $\bx_k^{[s]} \in\mathbb{R}^m$ is an uncoded  function of the local gradient $\nabla f_k\big({\bm \theta}^{[s]}\big)$, and $\bz^{[s]} $ is channel noise i.i.d. according to distribution $\mathcal{N}(0,{N_0}\bI_m)$.  
The transmit power constraint of a device is given as 
\begin{align}\label{eq: power constraint}
\|\bx_k^{[s]}\|^2\leq P,
\end{align}
accounting for per-block power constraints.

\subsection{Performance Metrics}\label{sec: metrics}
In this paper, we aim at designing a wireless federated learning protocol to implement Bayesian learning via Langevin MC under DP constraints. In this subsection, we formalize the performance criteria of interest and elaborate on the role of channel noise in achieving them. 
\subsubsection{Approximation Error}
Denoting  as $p^{[s]}(\bm \theta)$  the distribution of the sample  ${\bm \theta}^{[s]}$ at the $s$-th iteration, the quality of the sample is measured by the 2-Wasserstein distance between $p^{[s]}(\bm \theta)$  and the target global posterior $p({\bm\theta}|\cD)$ in \eqref{eq: posterior}. This is defined as \cite{dalalyan2017further,chatterji2018theory}
\begin{align} \label{eq: ws distance}
W_2\big(p^{[s]}(\bm \theta),p({\bm\theta}|\cD)\big)=\left(\inf_{p({\bm \theta}, {\bm \theta}')} \int_{\mathbb{R}^d\times\mathbb{R}^d} \l\|{\bm \theta}-{\bm \theta}'\r\|_2^2 p({\bm \theta}, {\bm \theta}') \mathrm{d} {\bm \theta} \mathrm{d} {\bm \theta}' \right)^{1/2},
\end{align}
where joint distribution $p({\bm \theta}, {\bm \theta}')$ is constrained to have marginals $p^{[s]}({\bm \theta})$  and $p({\bm \theta}'|\cD)$. The Wasserstein distance is a standard measure of discrepancy between two distributions, and it is routinely used for the analysis of MC algorithms (see, e.g., \cite{dalalyan2017further, chatterji2018theory}). It has some useful properties with respect to other measures such as the Kullback-Leibler divergence and total variation distance. For instance, it is well defined and informative even when the two distributions have disjoint supports \cite{arjovsky2017wasserstein}. 

\subsubsection{Differential Privacy} We consider a ``honest-but-curious" edge server that may attempt to infer information about local data sets from the received signals $\by$. 
 We impose the standard $(\epsilon,\delta)$-DP metric with some  $\epsilon>0$ and $\delta\!\in\![0,1)$, for each device $k$. This amounts to the inequalities
\begin{equation}\label{eq: def dp}
P (\{\by^{[s]}\}_{s=1}^{S} | \{\cup_{i\neq k}\cD_i\} \cup \cD_k') \leq \exp(\epsilon)P(\{\by^{[s]}\}_{s=1}^S  | \{\cup_{i\neq k}\cD_i\} \cup \cD_k'')+\delta, 
\end{equation}  
for all $k=1,\dots,K$,  where $S$ is the number of communication rounds, and $P (\by^{[s]} | \cD)$ represents the distribution of the received signal \eqref{eq: rev signal} conditioned on the global data set $\cD$. Condition \eqref{eq: def dp} must hold for any two possible neighboring data sets $\cD_k'$ and $\cD_k''$ differing only by one sample, i.e., $\|\cD_k'-\cD_k''\|_1=1$, and for all data sets $\{\cD_i\}_{i\neq k}$ of other devices. 

\subsubsection{On the Role of Channel Noise} Without the additive random term in ${\bm \xi}^{[s+1]}$ in \eqref{eq: LMC}, LMC coincides with standard gradient descent (GD) for frequentist learning, which was studied in \cite{liu2020privacy} in a federated setting under the DP constraints \eqref{eq: def dp}. 
 For convex models, the additive channel noise on the uplink channel \eqref{eq: rev signal} is  harmful to the convergence rate of GD with no privacy constraints \cite{liu2020privacy}, although it can improve the generalization performance \cite{yang2021revisiting}.  
 In \cite{liu2020privacy,seif2020wireless,koda2020differentially}, it was shown that channel noise can be repurposed to ensure the DP constraints~\eqref{eq: def dp} for values of $(\epsilon,\delta)$ that depend on the SNR level.  

In this paper, we observe that, in Bayesian learning via MC, by \eqref{eq: LMC},  noise can potentially contribute not only to privacy but also to MC sampling, without necessarily compromising the learning performance.  This idea was first introduced in \cite{liu2021channel} for one-shot MC sampling methods, and is studied here for the first time for iterative schemes. 
Specifically, we investigate the joint role of channel noise as a contributor to MC sampling and as a privacy-inducing mechanism. This interplay between MC sampling and DP was analyzed under ideal communication for centralized learning in \cite{wang2015privacy}, and the impact of channel noise in distributed settings is studied here for the first time. 


\subsection{Assumptions on the Log-Likelihood}
Finally, we list several standard assumptions we make on the global cost function $f(\bm \theta)$ in \eqref{eq: gb func} and on its gradient.

\begin{assumption}[Smoothness]\label{assumption: Lipschitz}\emph{
The global cost function $f(\bm \theta)$ is smooth with constant $L>0$, that is, it is continuously differentiable and the gradient $\nabla f(\bm \theta)$  is Lipschitz continuous with constant $L$, i.e.,   
\begin{align}\label{eq: Lipschitz} 
\|\nabla f({\bm \theta})-\nabla f({\bm \theta}')\| \leq L \| {\bm \theta}-{\bm \theta}'\|,  \quad \text{for all } {\bm \theta},{\bm \theta}' \in \mathbb{R}^m. 
\end{align} 
}
\end{assumption}

\begin{assumption}[Strong Convexity]\label{assumption: Strong Convexity}\emph{The global cost function $f(\bm \theta)$ is strongly convex, i.e., the following inequality holds for some constant $\mu>0$ 
 \begin{equation}\label{eq: PL ineq}
\l[\nabla f(\bm \theta)-\nabla f({\bm \theta}')\r]^{\sf T} ({\bm \theta}-{\bm \theta}') \geq \mu \|{\bm \theta}-{\bm \theta}'\|^2, \quad \text{for all } {\bm \theta},{\bm \theta}' \in \mathbb{R}^m. 
 \end{equation}
 }
\end{assumption}

 Assumptions \ref{assumption: Lipschitz} and \ref{assumption: Strong Convexity} imply the following inequality \cite[Lemma 3.11]{bubeck2014convex}    
  \begin{align} \label{eq: smooth and sc 2}
 \l[\nabla f(\bm \theta)-\nabla f({\bm \theta}')\r]^{\sf T} ({\bm \theta}-{\bm \theta}')  \geq \frac{\mu L}{\mu+L} \|{\bm \theta}-{\bm \theta}'\|^2  + \frac{1}{\mu+L} \l\|\nabla f(\bm \theta)-\nabla f({\bm \theta}') \r\|^2 ,  \nn\\   
  \text{for all } {\bm \theta},{\bm \theta}' \in \mathbb{R}^m.  
 \end{align}

 \begin{assumption}[Bounded Local Gradient]\label{assumption: BLL}\emph{The local gradient is bounded as     
  \begin{align}\label{eq: BV}
\l\|\nabla f_k({\bm \theta}^{[s]}) \r \| \leq \ell ,\quad  \text{for all }  k, s,
 \end{align} 
 and some constant $\ell>0$. 
 }
\end{assumption} 
 This last assumption is  essential to ensure DP requirements \cite{chen2020understanding,wang2015privacy}. One can choose the parameter of $\ell$ as the maximum value of $\|\nabla f_k({\bm \theta}^{[s]})  \| $, or, in practice, clipping the gradient as  $\overline{\nabla f_k}({\bm \theta}^{[s]}) \!= \! \min \{1, {\ell}/{\|{\nabla f_k}({\bm \theta}^{[s]})\|}\}{\nabla f_k}({\bm \theta}^{[s]})$. We will not account for the impact of  clipping in the analysis.


\section{Wireless Federated Langevin Monte Carlo (WFLMC)}\label{sec: WFLMC}
 In this section, we introduce wireless federated LMC (WFLMC). Specifically, we first present signal design and scheduling protocol, and then detail the rationale behind the proposed approach. The following sections will focus on the analysis of WFLMC. 
 
 \subsection{Signal Design and Scheduling Protocol}
 In each communication block, all devices transmit their local gradients simultaneously by using uncoded transmission of the form 
 \begin{align}\label{eq: tx signal}  
\text{(Transmitted Signal)}\quad \bx_k^{[s]}&= \alpha_k^{[s]}   \nabla f_k\big({\bm \theta}^{[s-1]}\big), 
\end{align}
for some power control parameter $\alpha_k^{[s]}$.
Specifically, we implement truncated channel inversion as in \cite{zhu2019broadband,zhu2020one,zhang2021turning}, whereby the power control parameter is selected as 
\begin{align}\label{def: pc}
\alpha_k^{[s]} = \left\{\begin{matrix}
{\alpha^{[s]} }/{h_k^{[s]}} , \quad   & \text{if} \ |h_k^{[s]}| \geq g^{[s]} ,\\ 
0, \quad & \text{if} \  |h_k^{[s]}|<g^{[s]},
\end{matrix}\right.
\end{align}
where  gain parameter $\alpha^{[s]}>0$ and threshold $ g^{[s]}>0$ are parameters to be optimized.  Accordingly, a device $k$ transmits only if its channel is large enough, i.e., $|h_k^{[s]}| \geq g^{[s]}$. {\color{black}We note that this channel-aware scheduling aims to avoid deep fading channels. Other
alternative scheduling policies may achieve better performance by considering also the importance of local data set, such as \cite{ren2020scheduling,liu2020data}. } 
We denote as $K_a^{[s]}$ the number of transmitting devices at round $s$, and $\mathcal{K}_a^{[s]}$ the corresponding set of transmitting~devices.  

To estimate the LMC update term \eqref{eq: LMC update term} in \eqref{eq: LMC}, the received signal \eqref{eq: rev signal}  is scaled and a Gaussian random vector is added. Specifically,  \emph{wireless federated LMC (WFLMC) update} is proposed  as 
\begin{align}\label{eq: LMC wireless}
{(\text {WFLMC})} \quad {\bm \theta}^{[s]}={\bm \theta}^{[s-1]} - \eta \frac{K}{\alpha^{[s]}K_a^{[s]}} \by^{[s]} +  \sqrt{\beta^{[s]}} {\bq}^{[s]},   
\end{align}
where the added noise 
${\bq}^{[s]}\sim \cN(0, \bI_m)$ is independent of all other variables. The variance $\beta^{[s]}$ of the noise term $\bq^{[s]}$ is chosen as a function of the learning rate $\eta$, channel noise $N_0$, number of active users $K_a^{[s]}$, and gain $\alpha^{[s]}$ as  
\begin{align}\label{eq: noise control} 
\beta^{[s]} = \max \bigg\{ 0, 2\eta- \frac{\eta^2N_0 K^2}{(\alpha^{[s]}K_a^{[s]})^2} \bigg\}.
\end{align}

\subsection{Understanding WFLMC}

{\color{black}To see the rationale behind the design \eqref{eq: LMC wireless}--\eqref{eq: noise control}, let us plug \eqref{eq: rev signal} and \eqref{eq: tx signal}--\eqref{def: pc} into \eqref{eq: LMC wireless}
to rewrite the WLMC update as 
\begin{align}
{\bm \theta}^{[s]}&={\bm \theta}^{[s-1]} - \eta \frac{K}{K_a^{[s]}} \sum_{k\in \cK_a^{[s]} } \nabla f_k\big({\bm \theta}^{[s-1]}\big)  \underbrace{-   \eta\frac{K}{\alpha^{[s]}K_a^{[s]}}\bz^{[s]}+ \sqrt{\beta^{[s]}} {\bq}^{[s]}}_{\text{effective noise } \widetilde{\bz}^{[s]}}\label{eq: LMC wireless rw} \\
&={\bm \theta}^{[s-1]} - \eta \bigg[\underbrace{\frac{K}{K_a^{[s]}} \sum_{k\in \cK_a^{[s]} } \nabla f_k\big({\bm \theta}^{[s-1]}\big) + \sqrt{\widetilde{\beta}^{(s)}}\Delta^{[s]}}_{\widehat {\nabla f}({\bm \theta}^{(s-1)})} \bigg]+\sqrt{2\eta} \widehat{\bm \xi}^{[s]}, \label{eq: LMC wireless gd noise} \end{align} 
where $\widehat{\bm \xi}^{[s]}\! \sim \! \cN(0, \bI_m)$ and $\Delta^{[s]} \!\sim \! \cN(0, \bI_m)$ are the i.i.d. sequences for $s=1,2,\dots$, and we~define 
\begin{align}\label{eq: rem noise}  
\widetilde{\beta}^{[s]}=\max \bigg\{ 0,   \frac{N_0K^2}{(\alpha^{[s] } K_a^{[s]})^2} -\frac{2}{\eta} \bigg\}.  
\end{align}  }

By \eqref{eq: LMC wireless rw},  the power scaling \eqref{def: pc} ensures that the gradients of the active devices sum at the receiver, and the scaling by $K/(\alpha^{[s]}K_a^{[s]})$ in \eqref{eq: LMC wireless} of the received signal compensates for the resulting multiplier of the sum-gradient. The term $ {K}/{K_a^{[s]}} \sum_{k\in \cK_a^{[s]} } \nabla f_k\big({\bm \theta}^{[s-1]}\big) $ is an empirical estimate of the gradient $ \sum_{k=1}^K\nabla f_k\big({\bm \theta}^{[s-1]}\big)$ in \eqref{eq: LMC} which is exact if $K_a^{[s]}=K$. 

Based on the discussion so far, in order for \eqref{eq: LMC wireless rw} to be an estimate of the LMC update \eqref{eq: LMC}, we should ideally ensure that the variance of the effective noise term $\widetilde{\bz}^{[s]}$ be equal to $2\eta$. However, the power of this term, namely, 
\begin{align}
\widetilde{\sigma}_z^2= \eta^2K^2/(\alpha^{[s]} K_a^{[s]})^2+ \beta^{[s]} 
\end{align} 
can only be partially controlled through power gain $\alpha^{[s]}$ and added noise variance $\beta^{[s]}$. In particular, we can always choose the added noise variance $\beta^{[s]}$ such that the variance $\widetilde{\sigma}_z^2$ is no smaller than $2\eta$. This condition is ensured by \eqref{eq: noise control}. In particular with \eqref{eq: noise control}, the effective noise in \eqref{eq: LMC wireless rw}  can be decomposed into two parts as indicated in \eqref{eq: LMC wireless gd noise}: 

\noindent 1) \emph{LMC noise:} The term $\sqrt{2\eta}\widehat{\bm \xi}^{[s]}$ with $\widehat{\bm \xi}^{[s]}\sim \cN(0, \bI_m)$ serves the role of LMC noise with variance $2\eta$; 

\noindent  2) \emph{Gradient estimation noise:} The remaining noise, denoted as $-\eta \sqrt{\widetilde{\beta}^{(s)}}\Delta^{[s]}$ with $\Delta^{[s]} \sim \cN(0, \bI_m)$, acts as a perturbation on the gradient estimate with variance $\widetilde{\beta}^{[s]}$ in \eqref{eq: rem noise}. Note that the variance $\widetilde{\beta}^{[s]}$ of estimate noise is non-zero if the channel noise power $N_0$ is large. 

 \section{Convergence and Privacy Analysis of WFLMC} \label{sec: analysis WFLMC}
 In this section, we focus on the performance analysis of WFLMC in terms of (i) convergence through the 2-Wasserstein distance as defined in \eqref{eq: ws distance}; and (ii) privacy under the DP criterion \eqref{eq: def dp}. In this section, we assume that the  sequence of power gain and scheduling threshold parameters $\{\alpha^{[s]}, g^{[s]}\}$ is fixed, and the results are given for an arbitrary sequence $ \{h_k^{[s]}\}$  of~channels. 

 \subsection{Convergence Analysis}\label{sec: convergence analysis}
 We now study the distribution $p^{[s]}({\bm \theta})$ of the sample ${\bm \theta}^{[s]}$ produced by WFLMC via the updates \eqref{eq: LMC wireless}. We recall that the goal of LMC is to produce samples distributed according to the global posterior $p(\cD|{\bm \theta})$.  As discussed in Sec. \ref{sec: metrics}, we measure the approximation error via 2-Wasserstein distance \eqref{eq: ws distance}. 
 
 There are three main contributions to the discrepancy between the distribution of $p^{[s]}({\bm \theta})$ and the target posterior $p({\bm \theta}|\cD)$: 
\begin{enumerate} 
\item the initial discrepancy, which is measured by the 2-Wasserstein distance  $W_2\big(p^{[0]}(\bm \theta),p({\bm\theta}|\cD)\big)$   between the initial distribution $p^{[0]}(\bm \theta)$ and the posterior $p({\bm\theta}|\cD)$;
\item the gradient error in \eqref{eq: LMC wireless gd noise}, namely  
\begin{align}\label{eq: grad error}
\text{(Gradient error)} \quad N_g^{[s]}={\nabla f}({\bm \theta}^{[s-1]})- \widehat {\nabla f}({\bm \theta}^{[s-1]}),
\end{align}
 which is caused by the excess noise $\sqrt{\widetilde{\beta}^{(s)}}\Delta^{[s]}$ and by the fact that only a subset of $K_a^{[s]}$ devices is active; 
\item  and the discretization error due to the approximation \eqref{eq: int. discre},  which is given as 
\begin{align}
\text{(Discretization error)}  \quad N_d^{[s]}= \int_{(s-1)\eta}^{s \eta}  \nabla f(\bar{\bm \theta}^{(t)}) - \nabla f(  \bar{\bm \theta}^{((s-1) \eta )} ) \mathrm{d}t.
\end{align}
\end{enumerate}
We now bound the last two terms separately, and then use these results to bound  square of the desired 2-Wasserstein distance $W_2\big(p^{[s]}(\bm \theta),p({\bm\theta}|\cD)\big)^2$ in \eqref{eq: ws distance}.


  \begin{lemma}[Upper Bound on the Gradient Error]\label{lemma: gradient error}
\emph{Under Assumption \ref{assumption: BLL}, for any $s$-th  communication round, the average power of the gradient error is bounded as 
\begin{align}\label{eq: gradient error}
\E\l[\| N_g^{[s]} \|^2 \r] \leq 4\ell^2 \big(K-K_a^{[s]}\big)^2+{\widetilde{\beta}^{[s]}}, 
\end{align}
where the expectation is taken with respect to the excess noise  $\sqrt{\widetilde{\beta}^{[s]} }\Delta^{[s]}$ and  the variance ${\widetilde{\beta}^{[s]}}$ is defined in \eqref{eq: rem noise}.
}
\end{lemma} 
\proof See  Appendix \ref{proof: gradient error}.

The error bound \eqref{eq: gradient error} is comprised of two parts. The first term is the estimation error due to scheduling, which is zero  if all the devices transmit in communication rounds, i.e., if $K_a^{[s]}=K$.  
The second term is the variance $\widetilde{\beta}^{[s]}$ of the excess noise $\sqrt{\widetilde{\beta}^{[s]} }\Delta^{[s]}$ in \eqref{eq: rem noise}.

The discretization error is constant for 
 for all communication round $s=1,2,\dots$, and can be bounded by following \cite[Lemma 3]{dalalyan2017further} as detailed in the next lemma. 
  \begin{lemma}[Upper Bound on the Discretization Error]\label{lemma: discr error}
\emph{For any  communication round $s$, the discretization error is invariant, and it is upper bounded as 
\begin{align}\label{eq: discr error}
\E\l[\| N_d^{[s]} \|^2 \r] \leq  \frac{\eta^4L^3 m}{3}+ \eta^3 L^2m, 
\end{align}
where the average is computed with respect to the joint distribution of the LDP \eqref{eq: LD} and  the~LMC~\eqref{eq: LMC}. 
}
\end{lemma} 
\proof The proof is detailed in Appendix \ref{proof: discr error}. 

{\color{black}Lemma \ref{lemma: discr error} shows that the discretization error grows with model dimension $m$, with the learning rate $\eta$, and with the smoothness constant $L$ of the global cost function $f(\bm \theta)$. Note that a larger $L$ implies a less smooth function $f(\bm \theta)$. }
 
We now leverage Lemma \ref{lemma: gradient error} and Lemma \ref{lemma: discr error} to bound the square of 2-Wasserstein distance. 

 \begin{proposition}[Bound for 2-Wasserstein Distance] \label{prop: convergence}\emph{For a learning rate $0<\eta\leq 2/L$, define 
 \begin{align}\label{def: gamma}
\gamma= \left \{ \begin{matrix}
1-\eta \mu , \quad   &0<\eta\leq 2/(\mu+L) ,\\ 
 \eta L-1, \quad & 2/(\mu+L)\leq \eta \leq 2/L.
\end{matrix}\right.
\end{align}
 Under Assumptions~\ref{assumption: Lipschitz}, \ref{assumption: Strong Convexity} and  \ref{assumption: BLL}, after any number $s'$ of iterations, the 2-Wasserstein distance between the sample distribution produced by WFLMC and the global posterior is upper bounded~as
\begin{align} \label{eq: result convergence}
W_2\big(p^{[s']}(\bm \theta),p({\bm\theta}|\cD)\big)^2  \leq \Big(\frac{1+\gamma}{2}\Big)^{2s'}W_2\big(p^{[0]}(\bm \theta),p({\bm\theta}|\cD)\big)^2 + \sum_{s=1}^{s'} \Big(\frac{1+\gamma}{2}\Big)^{2(s'-s)}  \frac{2(1+\gamma)}{1-\gamma} \nn\\
\times \bigg[ \frac{\eta^4L^3 m}{3}+ \eta^3 L^2m  + 4\eta^2\ell^2\big(K-K_a^{[s]}\big)^2 + \eta^2\widetilde{\beta}^{[s]} \bigg]  \\
\overset{\Delta}=\widetilde{W}_2\big(p^{[s']}(\bm \theta),p({\bm\theta}|\cD)\big)^2.
  \end{align}
   }
\end{proposition}
\begin{proof} The proof follows from the upper bound (see, e.g., \cite{dalalyan2017further})
 \begin{align} \label{eq: sq dis}
W_2\big(p^{[s]}(\bm \theta),p({\bm\theta}|\cD)\big)^2 \leq {\E\l[\l\| {\bm \theta}^{[s]}-  \bar{\bm \theta}^{(s\eta)}\r\|^2\r]},
 \end{align}
 where $\bar{\bm \theta}^{(s\eta)}$ is obtained from the LDP \eqref{eq: LD}, while ${\bm \theta}^{[s]}$ is the sample produced by the LMC \eqref{eq: LMC} with the LMC noise $\sqrt{2\eta}\widehat{ \bm \xi}^{[s]}=\sqrt{2}  \int_{s\eta}^{s\eta+\eta} \mathrm{d}\boldsymbol{B}^{(t)}$. Accordingly, the expectation is taken with respect to the Brownian motion $\boldsymbol{B}^{(t)}$ in \eqref{eq: LD} and over the initial  ${\bm \theta}^{(0)} \sim p^{[0]}({\bm \theta})$. Since, assuming the stationary of the LDP \eqref{eq: LD},  the marginal of the LDP output $\bar{\bm \theta}^{(s)}$ is the target distribution $p({\bm \theta}|\cD)$ for all $t>0$, the bound \eqref{eq: sq dis} follows from the definition \eqref{eq: ws distance} by upper bounding the infimum with the described choice of the joint distribution of ${\bm \theta}^{[s]}$ and $\bar{\bm \theta}^{(s\eta)}$. 

 Using \eqref{eq: LMC wireless gd noise} and \eqref{eq: diff LD}, the right hand side of \eqref{eq: sq dis} can be computed as  
  \begin{align}\label{eq: w2 mid}
& {\E\l[\l\| {\bm \theta}^{[s]}-  \bar{\bm \theta}^{(s\eta)}\r\|^2\r]}= \E\bigg[\Big\| {\bm \theta}^{[s-1]}- \bar{\bm \theta}^{( (s-1)\eta)} - \eta \big[\nabla f({\bm \theta}^{[s-1]} ) - \nabla f(  \bar{\bm \theta}^{((s-1)\eta)} )\big] \nn\\
& \qquad +\eta \big[ \underbrace{ {\nabla f}({\bm \theta}^{[s-1]})- \widehat {\nabla f}({\bm \theta}^{[s-1]}) }_{N_g^{[s]}} \big]+\underbrace{ \int_{(s-1)\eta}^{s \eta}  \nabla f(\bar{\bm \theta}^{(t)}) - \nabla f(  \bar{\bm \theta}^{((s-1) \eta )} )   \mathrm{d} t}_{N_d^{[s]}}   \Big\|^2\bigg]. 
\end{align} 
The rest of the proof involves applying the geometric inequality $2ab\leq \tau a^2+ \tau^{-1}b^2$ for any $\tau>0$, and using Lemma \ref{lemma: gradient error} and Lemma \ref{lemma: discr error} as detailed in Appendix \ref{proof: convergence}. 
\end{proof}

Proposition~\ref{prop: convergence} indicates that the 2-Wasserstein distance depends on the initial discrepancy $W_2\big(p^{[0]}(\bm \theta),p({\bm\theta}|\cD)\big)$, whose contribution decreases exponentially with $s'$; as well as the sum of contributions across the iteration index $s=1,2,\dots,s'$, with each $s$-th error term weighted down by a factor decreasing exponentially with $s'-s$, i.e., as one moves towards earlier iterations. This shows that the disturbances  at later communication rounds are more harmful to the approximation accuracy. Furthermore, the contribution of each iteration $s\geq 1$ is given by the sum of the gradient error bounded in Lemma~\ref{lemma: gradient error},  and the discretization error bounded in Lemma~\ref{lemma: discr error}. 

Another interesting aspect highlighted by the bound \eqref{eq: result convergence} concerns the optimal choice of the learning rate $\eta$. The upper bound \eqref{eq: result convergence} increases with $\gamma$ in \eqref{def: gamma}, and the minimum value of $\gamma$ is attained when $\eta=2/(\mu+L)$. However, a large learning rate $\eta$ causes the gradient error bound \eqref{eq: gradient error}  and the discretization error bound  \eqref{eq: discr error} to increase by Lemma \ref{lemma: gradient error} and Lemma \ref{lemma: discr error}. Thus, the optimal learning rate is in the range of $\eta \in (0, 2/(\mu+L)]$.


 \subsection{Differential Privacy Analysis}\label{sec: DP analysis}
 The WFLMC scheme implicitly implements a Gaussian DP mechanism \cite{dwork2014algorithmic,wang2015privacy}, since the channel noise $\bz^{[s]}$ in \eqref{eq: LMC wireless rw} is added to the disclosed function  $\sum_{k \in \cK_a^{[s]} } h_k^{[s]}\alpha_k^{[s]}  \nabla f_k\big({\bm \theta}^{[s-1]}\big) $. Note that the noise $\bq^{[s]}$ in \eqref{eq: LMC wireless rw}  is added by the edge server, and hence it does not contribute to privacy. Furthermore, while only part of the channel noise $\bz^{[s]}$ is useful for LMC (see \eqref{eq: LMC wireless gd noise}), the entire variance $N_0$ contributes to DP. 
 
 For the Gaussian mechanism, the privacy level $(\epsilon,\delta)$ depends on the sensitivity of the disclosed information and on the variance of the added noise \cite{dwork2014algorithmic,liu2020privacy}. In a manner consistent to the definition \eqref{eq: def dp} of DP, the sensitivity quantifies the maximum change of the disclosed function by replacing a single data point. The sensitivity for  device $k$ is accordingly defined as   
 \begin{align}
{(\text{Sensitivity})} \ \  \chi _k^{[s]}=\max_{\mathcal{D}'_k,\mathcal{D}''_k} \bigg\| h_k^{[s]} \alpha_k^{[s]}   \Big[\nabla f_k'\big({\bm \theta}^{[s-1]}\big)  - \nabla f_k''\big({\bm \theta}^{[s-1]}\big)\Big]\ \bigg\|,
\end{align}
where $\nabla f_k'\big({\bm \theta}^{[s-1]}\big)$ and $\nabla f_k''\big({\bm \theta}^{[s-1]}\big) $ are computed by using data sets $\mathcal{D}'_k$ and $\mathcal{D}''_k$ respectively, and we have $\|\mathcal{D}'_k-\mathcal{D}''_k\|_{1}=1$. By the triangular inequality and Assumption~\ref{assumption: BLL}, we plug in the definition of $\alpha_k^{[s]}$ \eqref{def: pc}  and have the bound
\begin{align}\label{eq: bound def sensitivity}   
\chi_k^{[s]} \leq {\bf 1}\big[|h_k^{[s]}|\geq g^{[s]}\big] \cdot {2 \alpha^{[s]}}\ell.   
\end{align}
Following \cite[Lemma 1]{liu2020privacy}, one can interpret the ratio  $(\chi_k^{[s]} )^2/N_0$ as the privacy loss in each communication rounds. This is formalized in the following proposition.
 \begin{proposition}[DP Guarantees]\label{lemma: privacy constraint}\emph{ For any given sequence of parameters $\{\alpha^{[s]}, g^{[s]}\}$ and channels $\{h_k^{[s]}\}$, after $s$ communication rounds, WFLMC guarantees $(\epsilon,\delta)$-DP   if the following condition is satisfied 
 \begin{align}
\sum_{s=1}^{s'} \frac{{\bf 1}\big[|h_k^{[s]}|\geq g^{[s]}\big] \cdot 2(\alpha^{[s]}   \ell)^2}{N_0} &\leq \l(\sqrt{\epsilon+\l[\mathcal{C}^{-1}\l({1}/{\delta}\r)\r]^2}-\mathcal{C}^{-1}\l({1}/{\delta}\r)\r)^2  \nn \\
&\overset{\Delta}{=}\mathcal{R}_{\sf dp}(\epsilon,\delta), \ \text{for all } k, \label{eq: privacy constraint} 
\end{align}
where $\mathcal{C}^{-1}(x)$ is the inverse function of $\mathcal{C}(x)=\sqrt{\pi}xe^{x^2}$, and  $ {\bf 1}[\cdot]$ is the indicator function.
}
\end{proposition}
\begin{proof}The result follows from \cite[Lemma 1]{liu2020privacy}, although reference \cite{liu2020privacy} did not account for the threshold-based scheduling in \eqref{def: pc}. The extension is direct by redefining the effective channel gain as $\alpha^{[s]} \cdot {\bf 1}\big[|h_k^{[s]}|\geq g^{[s]}\big]$. 
\end{proof}

In accordance to the discussion above, the left-hand side of \eqref{eq: privacy constraint} quantifies the overall privacy loss across $s'$ rounds. Importantly, as anticipated, by \eqref{eq: privacy constraint}  the channel noise power $N_0$ contributes in full to the DP performance. In contrast, by \eqref{eq: LMC wireless gd noise}, only a portion of the channel noise contributes, in general, to the LMC update.   



\section{Optimal Power Allocation and Scheduling}\label{sec: optmization}
 
 In this section, we leverage Proposition \ref{prop: convergence} and Proposition \ref{lemma: privacy constraint}  to address the problem of minimizing the convergence error under the $(\epsilon, \delta)$-DP constraint \eqref{eq: privacy constraint} and the power  constraints \eqref{eq: power constraint} over power gain parameters and thresholds $\{\alpha^{[s]},g^{[s]}\}_{s=1}^{S}$ in \eqref{def: pc}. We recall that WFLMC carries out $S_b$ communication rounds for the burn-in period, which are followed by $S_u$ additional rounds to obtain the samples for use in downstream applications. The learning objective is to maximize the quality of  the last $S_u$ samples under the mentioned privacy and power constraints, which apply for the total of $S=S_b+S_u$ communication rounds.

Throughout this section, we assume that the sequence of channels $\{\{h_k^{[s]}\}_{k=1}^K\}_{s=1}^S$ is known in advance in order to enable optimization. {\color{black}This assumption can be relaxed at the cost of additional communication overhead.} Extensions to an online approach  can be directly obtained by iterative one-step-ahead optimization based on predicted values for the future parameters $\{\{h_k^{[s']}\}_{k=1}^K\}_{s'=s+1}^S$ as detailed in \cite{liu2020privacy}. The optimization is conducted at the edge server. To control AirComp transmission, the edge server broadcasts the optimized $\{\alpha^{[s]}, g^{[s]}\}$ at each communication round, which consumes a negligible amount of communication resources as compared with model broadcasting.

 The  joint power allocation and scheduling problem of interest is  formulated as  the min-max optimization   
 \begin{subequations}\label{generic OPT} 
 \begin{align} 
   \min_{{\{\alpha^{[s]}, g^{[s]}\}}_{s=1}^{S}}\max_{s'\in[S_b+1,  S_b+S_u]}&\quad \widetilde{W}_2\big(p^{[s']}(\bm \theta),p({\bm\theta}|\cD)\big)^2  \label{opt obj} \\
{\rm s.t.} \quad &\sum_{s=1}^{S} \frac{{\bf 1}[|h_k^{[s]}|\geq g^{[s]}] \cdot 2 (\alpha^{[s]}\ell)^2}{N_0} \leq \mathcal{R}_{\sf dp}(\epsilon, \delta)  \label{opt dp}\\
& {\bf 1}\big[|h_k^{[s]}|\geq g^{[s]}\big]\cdot \frac{(\alpha^{[s]}\ell)^2}{|h_k^{[s]}|^2} \leq P , \quad \forall k,s=1,\cdots, S .  \label{opt power constraint} 
 \end{align}
\end{subequations}
 The maximization in \eqref{opt obj} aims at ensuring that the  worst-case 2-Wasserstein distance is minimized across all $S_u$ samples after the burn-in period, under the DP constraint \eqref{opt dp} and the power constraint \eqref{opt power constraint}.  With its focus on the distribution of model parameters, problem \eqref{generic OPT} is notably distinct from the optimization problem for frequentist learning in \cite{liu2020privacy}, which only considers the quality of a single vector of model parameters in terms of training loss. 
 

The problem \eqref{generic OPT} is non-convex since the objective function is non-differentiable in the thresholds $\{g^{[s]}\}_{s=1}^S$.  In fact, by \eqref{def: pc}, the threshold $g^{[s]}$ affects the objective function through the number $K_a^{[s]}$ of active users. To make progress, we fix the scalar thresholds $g^{[s]}$ and optimize over the power gain parameter $\alpha^{[s]}$. 



 \subsection{Zero Additive Noise is Optimal}
 We start by simplifying problem \eqref{generic OPT} through the following observation. 
   \begin{lemma}[Zero Additive Noise]\label{lemma: zero add nosie}
\emph{Without compromising optimality, the variance of additive noise $\sqrt{\beta^{[s]}} {\bq}^{[s]}$ in \eqref{eq: LMC wireless} can be set as $\beta^{[s]}=0$, which is equivalent to imposing the following constraint on the power gain parameters 
\begin{align}\label{eq: LMC noise req}
\text{(LMC noise requirement)}\quad \alpha^{[s]}\leq\frac{K}{K_a^{[s]}}\sqrt{\frac{\eta N_0}{2}}, \quad \text{for } s=1,\cdots,S.
\end{align}
}
\end{lemma} 
\begin{proof} Assume by contradiction that we had $\beta^{[s]}>0$ at an optimal solution. By \eqref{eq: noise control}, this would imply that the optimal $\alpha^{[s]}$ satisfies the inequality ${\eta^2N_0K^2}/{(\alpha^{[s]} K_a^{[s]})^2} -{2}{\eta}<0$. But one can always choose the smaller value $\alpha^{[s]}=\sqrt{\eta N_0/2} K/K_a^{[s]} $, which achieves the same value of the objective function, while reducing the left-hand sides of the privacy and power constraints \eqref{opt dp}-\eqref{opt power constraint}. This yields a contradiction, completing the proof. 
\end{proof} 

The inequality \eqref{eq: LMC noise req} distinguishes two distinct regimes of operation of the system. If the equality \eqref{eq: LMC noise req}  is active, the gradient estimation noise power $\widetilde{\beta}^{[s]}$ in \eqref{eq: rem noise}  is zero, and hence the channel noise contributes in full to the LMC updates. In contrast, when the inequality is strict, only a fraction of the channel noise is useful for LMC, and the rest contributes to the gradient estimation noise power $\widetilde{\beta}^{[s]}$. 

\subsection{Optimization of Power Gain Parameters: Single-Sample and Constant Channels}\label{sec: single sample case}
We now tackle the optimization \eqref{generic OPT} over power gain parameters $\{\alpha^{[s]}\}_{s=1}^S$ for the special case $S_u=1$, and with constant channels and scheduling thresholds
\begin{align}
h_k^{[s]}&=h_k, \  \text{for} \ s=1,\dots, S, \label{eq: static channel}\\  
g^{[s]}&=g,  \ \text{for} \ s=1,\dots, S. \label{eq: static threshold}
 \end{align}
 The solution in this special case will turn out to be especially insightful, and the more general problem will be studied in the next subsection. Under theses simplifying assumptions, and using Lemma \ref{lemma: zero add nosie}, the min-max optimization in \eqref{generic OPT} reduces to the problem 
\begin{subequations}\label{os wt OPT1}
 \begin{align}
    \min_{{\{\alpha^{[s]}\}}_{s=1}^{S_b+1}} &\quad  \sum_{s=1}^{S_b+1} \Big(\frac{1+\gamma}{2}\Big)^{-2s}  \max \bigg\{0, \frac{\eta^2N_0K^2}{(\alpha^{[s] } K_a)^2} -{2}{\eta} \bigg\} \label{os wt obj} \\
{\rm s.t.} \quad &\sum_{s=1}^{S_b+1} \frac{ 2 (\alpha^{[s]}\ell)^2}{N_0}\leq \mathcal{R}_{\sf dp}(\epsilon, \delta),   \label{os wt dp}\\
& \frac{({\alpha^{[s]}\ell})^2}{|h_k|^2} \leq P , \quad \forall k \in \cK_a, s=1,\dots, S   \label{os wt power constraint}  \\
& \alpha^{[s]} \leq \frac{K}{K_a} \sqrt{\frac{\eta N_0}{2}} , \quad \forall k \in \cK_a, s=1,\dots, S.   \label{os eta noise} 
 \end{align}
\end{subequations}

\begin{theorem}\label{thm: optimal solution single sample} 
\emph{Under Assumptions \ref{assumption: Lipschitz}-\ref{assumption: BLL},  and assuming static channels and thresholds as in \eqref{eq: static channel}--\eqref{eq: static threshold},  the optimal solutions of problem \eqref{os wt OPT1} depends on power $P$ and learning rate $\eta$ according to the three regimes  illustrated in Fig. \ref{Fig: partition}, which are detailed as follows. 
 \begin{enumerate}
\item \emph{LMC-limited Regime:} If the condition 
\begin{align}\label{eq: cond eta}
\eta\leq\frac{K_a^2}{\ell^2K^2}\min\Big\{ \frac{ \mathcal{R}_{\sf dp}(\epsilon, \delta)} {S}, \min_{k\in\cK_a}\frac{2P|h_k|^2 }{N_0}\Big\}
\end{align}
 holds, the optimal power gain parameter is given as 
\begin{align}\label{eq: alpha LMC}
\alpha^{[s]}_{\sf opt}= \frac{K}{K_a}\sqrt{\frac{\eta N_0}{2}}, \ s=1,\dots, S.
\end{align}
\item \emph{Power-limited Regime:} If the condition 
\begin{align}\label{eq: cond P} 
{P}\leq \frac{N_0}{2\min_{k\in \cK_a}   |h_k|^2}\min\Big\{ \frac{ \mathcal{R}_{\sf dp}(\epsilon, \delta)} {S}, \frac{\ell^2 K^2\eta}{K_a^2 }\Big\}
\end{align} 
holds, the optimal power gain parameter is given as 
\begin{align}\label{eq: alpha power} 
\alpha^{[s]}_{\sf opt}= \min_{k\in \cK_a} \frac{\sqrt{P }|h_k|}{\ell} , \ s=1,\dots, S.
\end{align}
\item \emph{DP-limited Regime:} Otherwise, we have  optimal solution  
\begin{align}\label{eq: alpha DP}  
\alpha^{[s]}_{\sf opt}= \min\bigg\{ \Big(\frac{1+\gamma}{2}\Big)^{-s/2}\sqrt{\frac{\eta K N_0^{1/2}}{K_a \lambda^{1/2}}}, \min_{k\in \cK_a} \frac{\sqrt{P }|h_k|}{\ell} ,\frac{K}{K_a}\sqrt{\frac{\eta N_0}{2}} \bigg\},  
 \end{align}
 where the value of $\lambda$ can be obtained by bisection to satisfy the condition $\sum_{s=1}^{S_b+1}(\alpha_{\sf opt}^{[s]})^2=\frac{ N_0\mathcal{R}_{\sf dp}(\epsilon, \delta)}{2 \ell^2}$ for the active devices $k\in \cK_a $.   
  \end{enumerate}}
\end{theorem} 
\proof The problem is seen to be convex by changing variables $(\alpha^{[s]})^2=a^{[s]} \geq 0$ and the solution approach involves applying Lagrange multiplier method and  Karush-Kuhn-Tucker (KKT) conditions in a manner similar to \cite[Theorem 1]{liu2020privacy}. Details can be found in Appendix \ref{proof: optimal solution single sample}. 
 \begin{figure}[t]
\centering
\includegraphics[width=8cm]{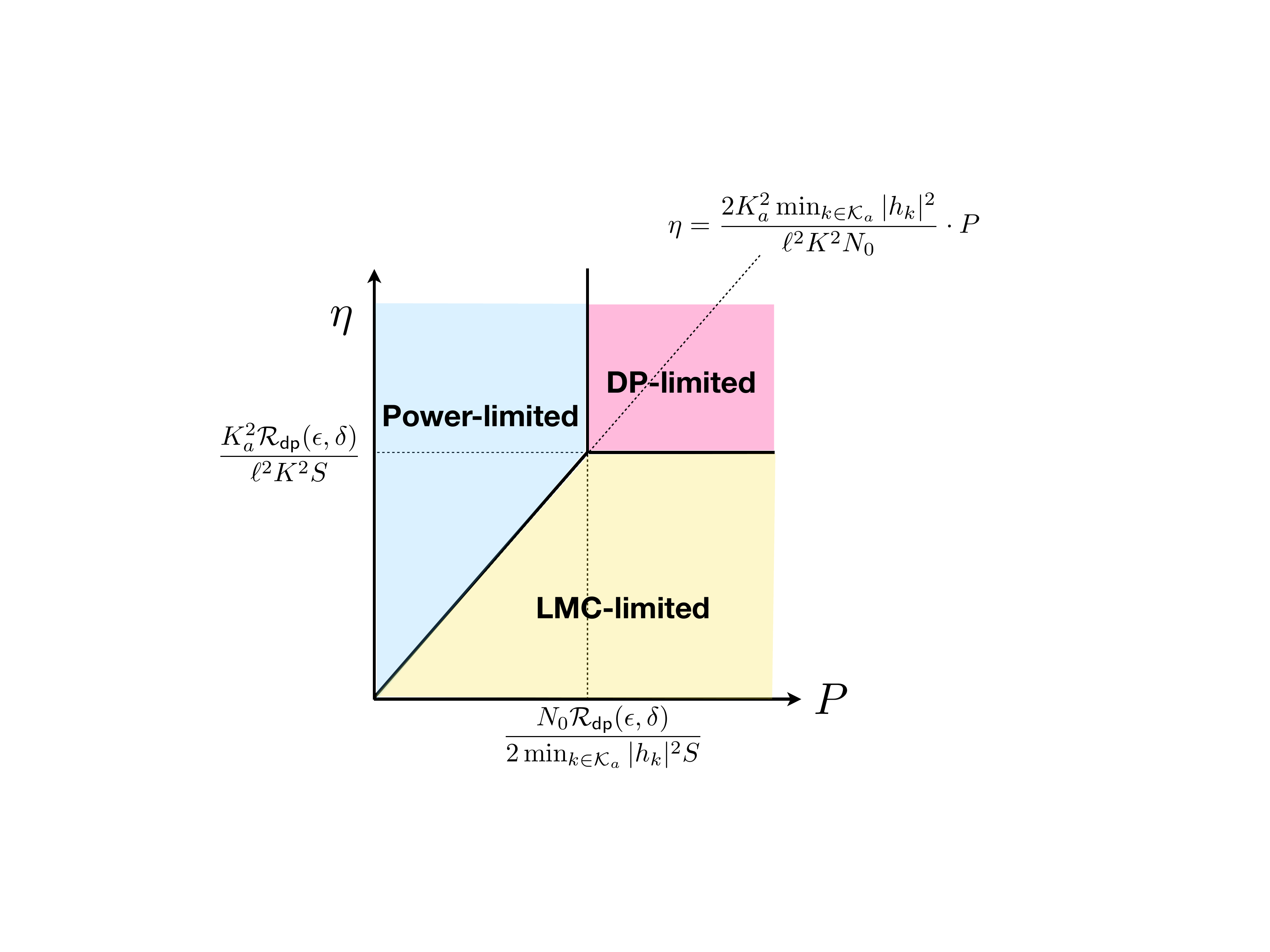}
\vspace{-6mm}
\caption{Illustration of the different regimes that describe the optimal power control strategy in Theorem \ref{thm: optimal solution single sample}.}
\label{Fig: partition}
\vspace{-6mm}
\end{figure}

The three regimes highlighted in Fig. \ref{Fig: partition} correspond to settings where each of the corresponding constraints, i.e.,  \eqref{eq: alpha LMC}, \eqref{eq: alpha power} or \eqref{eq: alpha DP}, are active. Accordingly, we have the following~observations: 

\noindent $\bullet$ In the LMC-limited regime, the channel noise contributes in full to LMC update, and hence the presence of channel noise and DP does not affect the performance of LMC. In contrast, in the other two regimes, the additional gradient noise causes a loss as compared to a noiseless implementation of LMC. 

\noindent $\bullet$ In the DP-limited regime, the privacy constraint affects the performance of WFLMC, while in the other regimes, privacy is obtained ``for free", i.e., as a direct consequence of the presence of channel noise.  
To optimize the learning performance, the power control parameter should be adaptive across the iterations.

\subsection{Optimization of Power Gain Parameters: General Case}
In this subsection, we consider the general optimization problem \eqref{generic OPT} with any $S_u\geq 1 $ and under time-varying channels. 
 To this end, we include LMC noise requirement leveraging Lemma~\ref{lemma: zero add nosie}, and start by rewriting the min-max optimization \eqref{generic OPT} in the epigraph form 
\begin{subequations}\label{ms wt OPT1}
 \begin{align}
{(\text {\bf Multi-Sample Opt.})} \    \min_{{\{\alpha^{(s)}\}}_{s=1}^{S}} &\quad  \nu \label{ms wt obj} \\
{\rm s.t.} \quad &\sum_{s=1}^{S} \frac{{\bf 1}[|h_k^{[s]}|\geq g^{[s]}] \cdot 2 (\alpha^{[s]}\ell)^2}{N_0}\leq \mathcal{R}_{\sf dp}(\epsilon, \delta), \quad \forall k  \label{ms dp}\\
& \l(\frac{{\bf 1}\big[|h_k^{[s]}|\geq g^{[s]}\big] \alpha^{[s]}\ell}{h_k^{[s]}}\r)^2 \leq P , \quad \forall k,s   \label{ms power constraint} \\
&\alpha^{[s]} \leq \frac{K}{K_a^{[s]}} \sqrt{\frac{\eta N_0}{2}} , \quad \forall s   \label{ms LMC noise} \\
& \widetilde{W}_2\big(p^{[s']}(\bm \theta),p({\bm\theta}|\cD)\big)^2  \leq \nu,  \ \forall s'=S_b+1, \dots, S.   \label{ms ws}  
 \end{align}
\end{subequations}
Having made the change of variables $(\alpha^{[s]})^2=a^{[s]} \geq 0 $, problem \eqref{ms wt OPT1} can be easily seen to be convex. Therefore, it can be solved using standard numerical tools. 

\subsection{Optimization of Truncated Thresholds}
  In this subsection, we turn to the problem of optimizing  \eqref{opt obj} over the thresholds $\{g^{[s]}\}_{s=1}^S$. As mentioned, this problem is characterized by a non-differentiable objective function, and it should be addressed jointly with the optimization of the power gain parameters $\{\alpha^{[s]}\}_{s=1}^S$. To make progress, we propose a sub-optimal approach that decouples the two problems by setting  the value of $\alpha^{[s]}$ to ensure equality in the constraint \eqref{opt power constraint} as 
  \begin{align}\label{eq: fpt}
\text{(Full Power Transmission)} \quad \alpha^{[s]}=\min_{k\in \cK_a^{[s]}} \frac{\sqrt{P }h_k^{[s]}}{\ell}, \quad s=1,\dots,S. 
\end{align} 
Note that the choice in \eqref{eq: fpt} is only made for the purpose of optimizing the thresholds. After the thresholds $\{g^{[s]}\}_{s=1}^S$ are optimized as explained next, the power gain parameters $\{\alpha^{[s]}\}_{s=1}^S$ are selected by following the previous subsections. 

Using \eqref{eq: fpt}, the min-max problem \eqref{opt obj} over $\{g^{[s]}\}_{s=1}^S$ can be expressed as  the sum of  gradient error over $s'$ communication rounds,  for $s'=S_b+1,\dots, S$
\begin{align}
 \min_{\{g^{[s]}\}_{s=1}^S} \quad \sum_{s=1}^{s'}  \Big(\frac{1+\gamma}{2}\Big)^{2(s'-s)}  \frac{2(1+\gamma)}{1-\gamma} \bigg[ 4\eta^2\ell^2\big(K-K_a^{[s]}\big)^2 + \widetilde{\beta}^{[s]} \bigg], \  \text{for} \  s'\in[S_b+1,  S]  
  \label{eq: min grad error}
\end{align}
where we have dropped the constraints since they are assumed to be dealt with by the subsequent optimization over $\{\alpha^{[s]}\}_{s=1}^S$. The minimization \eqref{eq: min grad error} can be addressed as $s'$ parallel optimizations, and it is equivalent to focus on $s'=S$ as 
\begin{align}\label{eq: min error g}
 \min_{g^{[s]}}\quad 4\ell^2 \big(K-K_a^{[s]}\big)^2+\max \bigg\{ 0,   \frac{N_0K^2\ell^2}{P( K_a^{[s]})^2 \min_{k\in\cK_a^{[s]}}|h_k^{[s]}|^2} -\frac{2}{\eta} \bigg\}, \  \text{for} \  s=1,\dots, S, 
\end{align} 
where the second term is obtained by the definition of $\widetilde{\beta}^{[s]}$ in \eqref{eq: rem noise}. {\color{black} The first term is estimation error depends on the bound $\ell$  of the norm of the local gradient as per Assumption 3. We note that this result can be directly extended to device-dependent bounds $\ell_k$ for each device $k$. With this extension, the threshold ${g^{[s]}}$ in the scheduling rule \eqref{def: pc} would depend on the properties of the local data set via constants $\{\ell_k\}_{k=1}^K$.}

{\color{black}
The threshold $g^{[s]}$ affects the objective function \eqref{eq: min error g}  through the number $K_a^{[s]}$ of active users. Decreasing $K_a^{[s]}$  aggravates the estimation error (first term) due to partial scheduling, while providing a chance to alleviate the excess noise (second term) by silencing the devices with the worst channels. The objective function \eqref{eq: min error g} can take at most $K$ different values that are attained by setting $g^{[s]}$ as one of the channel gains $ |h_1^{[s]}|, \cdots,  |h_K^{[s]}|$. Therefore, we can limit the search to these values without loss of optimality. It follows that the optimal solution for each threshold $g^{[s]}$ can be obtained by exhaustive search over the $K$ values $\{|h_k^{[s]}|\}_{k=1}^K$ by minimizing~\eqref{eq: min error g}. }

 \section{Numerical Results}\label{sec:simulation}
 
{\color{black}In this section, we investigate the effectiveness of the proposed scheme, WFLMC, as a mechanism to implement LMC on wireless channels using numerical experiments. We emphasize that our goal is not that of comparing the performance of Bayesian and frequentist techniques. This is a subject that has been extensively explored in the literature (see Sec. I), and we consider it to be beyond the scope of this contribution. Rather, our focus is on evaluating the effectiveness of the specific proposed implementation mechanism of LMC as compared to more conventional solutions. To this end, we consider the following benchmark schemes.}
\begin{itemize}
 \item [1)] {\bf WFMLC with equal power allocation:} This reference scheme follows the approach in \cite{seif2020wireless} of dividing up the DP constraint equally across all communication rounds. By Lemma~\ref{lemma: privacy constraint}, this corresponds to imposing the constraint  
\begin{align}\label{eq: equal dp}
\frac{2 (\alpha^{[s]} \ell)^2}{N_0}<\frac{\mathcal{R}_{\sf dp}(\epsilon,\delta)}{ \sum_{s=1}^S {\bf 1}[|h_k^{[s]}|\geq g^{[s]}]}
\end{align}
for each communication round. Condition \eqref{eq: equal dp}, along with the power constraint \eqref{opt power constraint} and  LMC noise requirement \eqref{eq: LMC noise req} yield the power scaling gains 
\begin{align}
 \alpha^{[s]}=\min\bigg\{ \frac{1}{\ell}\sqrt{ \frac{{N_0\mathcal{R}_{\sf dp}(\epsilon,\delta)}}{  {2 \sum_{s=1}^S {\bf 1}[|h_k^{[s]}|\geq g^{[s]}]}}}, \min_{k\in \cK_a^{[s]}} \frac{\sqrt{P }|h_k^{[s]}|}{\ell},\frac{K}{K_a^{[s]}}\sqrt{\frac{\eta N_0}{2}}\bigg\}, \ \forall s . 
\end{align}
 \item [2)]  {\bf WFLMC without DP constraint:} In this scheme, the power gain parameter $\alpha^{[s]}$ is set~as 
 \begin{align}
 \alpha^{[s]}=\min\bigg\{ \min_{k\in \cK_a^{[s]}} \frac{\sqrt{P }|h_k^{[s]}|}{\ell},\frac{K}{K_a^{[s]}}\sqrt{\frac{\eta N_0}{2}}\bigg\}, \ \forall  s ,  
\end{align}
  which corresponds to the optimal solution without the DP constraint. 
\end{itemize}
We also consider for reference a scenario characterized by ideal communication without fading (i.e., $h_k=1$), and channel noise (i.e., $N_0=0$). To guarantee privacy and implement LMC update, noise is added at each devices before transmission as 
\begin{align}
 \bx_k^{[s]}=   \nabla f_k\big({\bm \theta}^{[s-1]}\big)+ \sqrt{\sigma^{[s]}}\bn_k^{[s]}, 
\end{align}
and $\bn_k^{[s]}\sim \cN(0, \bI_m)$. The variance $\sqrt{\sigma^{[s]}}$ depends on the DP constraint. Specifically, following the analysis in WFLMC, we can define the following two additional benchmark schemes. 
 \begin{itemize} 
 \item [3)]    {\bf Noiseless federated LMC with DP constraint:} By solving problem \eqref{generic OPT} without power constraint \eqref{opt power constraint}, the optimal value $(\alpha^{[s]})_{\sf opt'}$ is used to set the variance of additive noise as $\sigma^{[s]}=N_0/[K(\alpha^{[s]})_{\sf opt'}^2]$ in order to satisfy the DP constraint. 
 \item [4)]    {\bf Noiseless federated LMC without DP constraint:} Removing the DP constraints, the implementation of the LMC update \eqref{eq: LMC} requires the variance of update noise in \eqref{eq: LMC update term} to equal $2\eta$, and by \eqref{eq: LMC wireless rw} we have $\sigma^{[s]}=2/(K\eta)$.
\end{itemize} 


As for the learning model, we consider a Gaussian linear regression with likelihood
 \begin{align}\label{eq: gl model}
p(v_n|{\bm \theta}, \bu_n) =  \frac{1}{\sqrt{2\pi}}e^{-\frac{1}{2} (v_n-{\bm \theta}^{\sf T}\bu_n)^2}, 
\end{align}
and the prior $p({\bm \theta})$ is assumed to follow Gaussian distribution $\cN(0, \bI_m)$. Therefore, the posterior $p({\bm \theta}|\cD)$ is the Gaussian $\cN\big((\bU\bU^{\sf T}+\bI)^{-1}\bU\bv,(\bU\bU^{\sf T}+\bI)^{-1}\big)$,  where $\bU=[\bu_1,\cdots,\bu_N]$ is the data matrix and $\bv\!\!=\!\![v_1, \cdots, v_N]^{\sf T}$ is the label vector. The strong convexity parameter $\mu$ and smoothness parameter $L$ are computed as the smallest and largest eigenvalues of the data Gramian matrix $\bU\bU^{\sf T}+\bI$. We will also consider experiments with the MNIST data set at the end of this section. 

We consider a synthetic data set $\{\bd_n \!  = \!   (\bu_n,v_n)\}_{n=1}^N$ with $N \!=\!  1200$ following the model \eqref{eq: gl model} with covariates $\bu_n   \in   \mathbb{R}^{m}$ drawn i.i.d. from Gaussian distribution $\mathcal{N}(0, \bI_m) $ where $m \! =\! 5$, and ground-truth model parameter ${\bm \theta}^*   \!\!= \!\! [0.071,
-0.518, 0.9342, 0.7198, 0.4676]^{\sf T} $.  
 The Wasserstein distance between two Gaussian distribution $p_1 \!\! =\!   \cN({\bf m}_1,   \bC_1\!)$ and $p_2 \! = \!   \cN({\bf m}_2,   \bC_2)$ is computed~as~\cite{sun2018gaussian} 
 \begin{align}\label{eq: gauss ws eval} 
 W_2(p_1,p_2)^2= \|{\bf m}_1-{\bf m}_2\|^2+\trace\Big(\bC_1+\bC_2-2(\bC_2^{1/2}\bC_1\bC_2^{1/2})^{1/2}\Big).
 \end{align}
 
 {\color{black}Unless stated otherwise, we consider homogeneous local data distributions by dividing the data set equally among the $K=30$ devices.  In this regard, we observe that Bayesian learning is generally not impacted by ``non-i.i.d." data distributions in the sense that global posterior is well defined irrespective of data homogeneity or heterogeneity at the devices. That said, in the presence of privacy constraint, the need  for clipping causes centralized and federated implementations of LMC to be different, and the performance loss of federated against centralized LMC becomes more pronounced for a smaller clipping threshold in the presence of heterogeneous distribution. Here we set the clipping threshold to $\ell=30$, and leave an investigation of the interplay between data heterogeneity and clipping to future work.} Furthermore, the initial sample ${\bm \theta}^{[0]}$ is drawn from prior with discrepancy $W_2\big(p^{[0]}(\bm \theta),p({\bm\theta}|\cD)\big)^2=6.64 $.

\subsection{Single-Sample Case}
We now focus on  single-sample case ($S_u=1$) by applying  the optimal power allocation detailed in Theorem~\ref{thm: optimal solution single sample} while scheduling all devices ($K_a^{[s]}=K$).  As in Theorem~\ref{thm: optimal solution single sample}, the channels $h_k^{[s]}$ are constant and set  to $0.01$ for all devices $k$; the number of communication burn-in periods is set to $S_b=50$; DP parameters are given as $\delta=0.01$ and $\epsilon=8$. 

 \begin{figure}
\centering    
\subfigure[$\eta=0.4/(\mu+L)$.]
{\includegraphics[width=7cm]{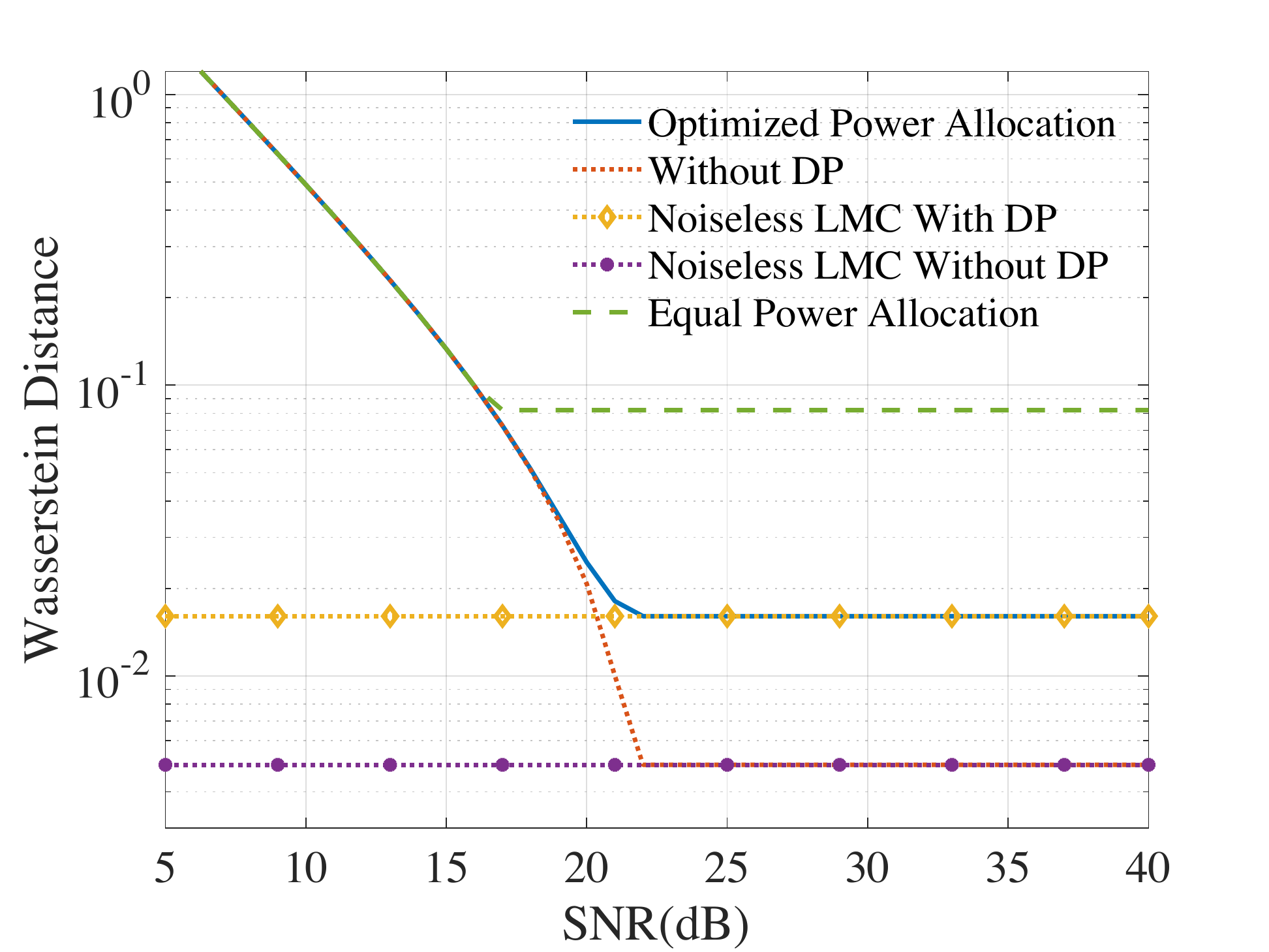}\label{Fig: SNR1}}
\subfigure[$\eta=0.13/(\mu+L)$.]
{\includegraphics[width=7cm]{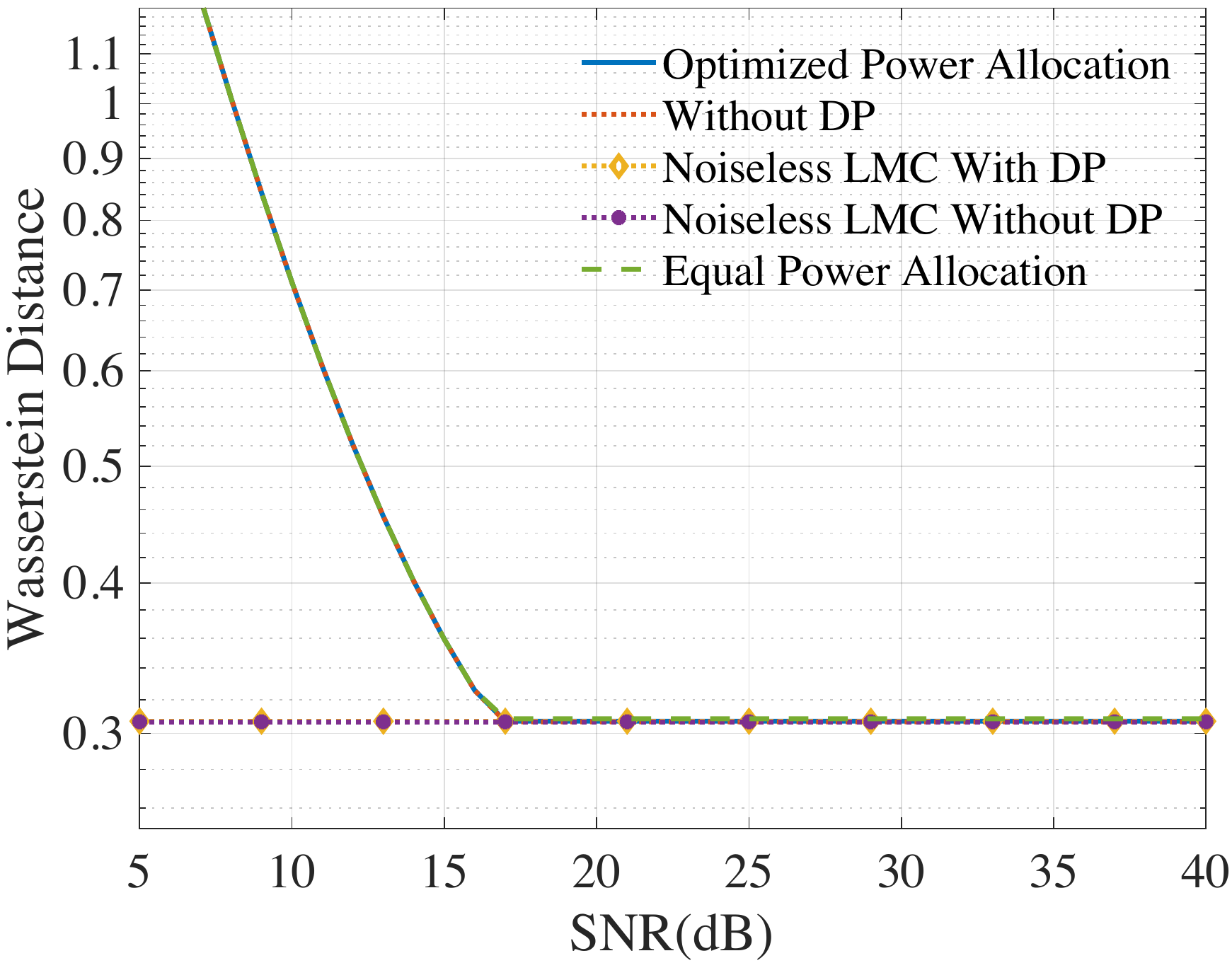}\label{Fig: SNR2}}
\caption{Bound \eqref{eq: result convergence} on the 2-Wasserstein distance  versus learning rate $\eta$ for WFLMC and noiseless LMC schemes ($\epsilon=8$, $\delta=0.01$,  $S_b=50$, $S_u=1$).}
\label{Fig: SNR}
\vspace{-6mm}
\end{figure}

\noindent $\bullet$ {\bf Impact of the SNR:} In Fig. \ref{Fig: SNR}, we plot the bound \eqref{eq: result convergence} on the 2-Wasserstein distance  versus SNR, defined as $P/(m N_0)$, by varying transmitted power $P$.  The learning rate is set as $\eta=0.4/(\mu+L)$ and $0.13/(\mu+L)$ in order to illustrate the different regimes defined in Fig. \ref{Fig: partition}.  In Fig. \ref{Fig: SNR1}, under the larger learning rate,  the performance is limited by the power constraint until $\SNR=16.6$ dB, after which it becomes limited by DP.  In the DP-limited regime, optimized power allocation attains better performance than static power allocation.  For $\SNR>22$ dB, the performance under optimized power allocation becomes equivalent to that of the system with ideal noiseless communication with the DP, showing that all channel noise is repurposed for DP mechanism while only partial is for MC sampling.  That is, imposing the DP constraint causes some performance loss due to the need of scaling down the transmission power to  decrease the effective SNR. In line with the results in Fig. \ref{Fig: partition}, Fig. \ref{Fig: SNR2} shows the transition from power-limited regime to LMC-limited regime for the smaller learning rate. In the LMC-limited regime, all the schemes have same performance since the channel noise variance is only determined by the LMC constraint.

\begin{figure}
\centering    
\subfigure[$\SNR=18$ dB.]
{\includegraphics[width=7cm]{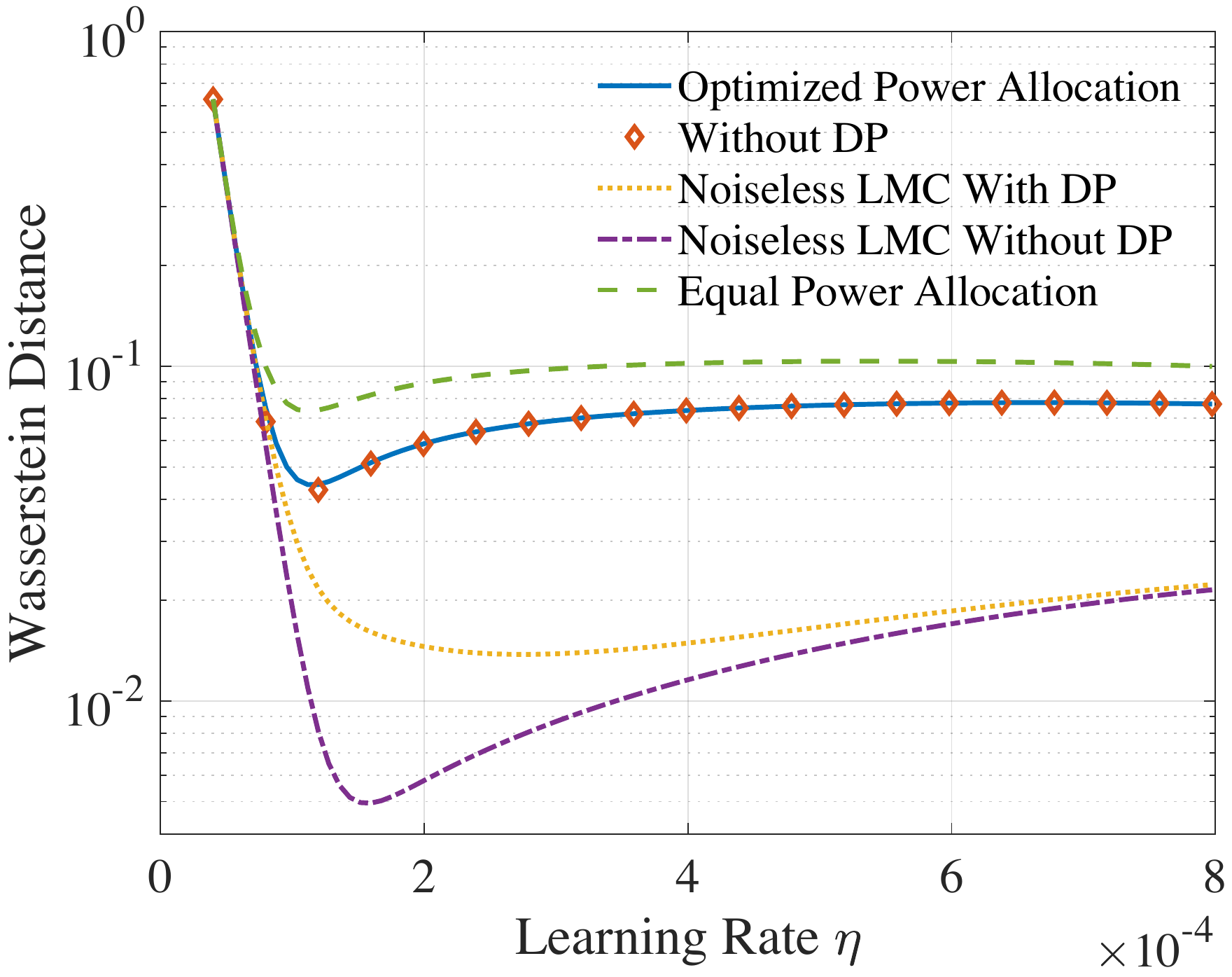}\label{Fig: etaSNR18}}
\subfigure[$\SNR=30$ dB.]
{\includegraphics[width=7cm]{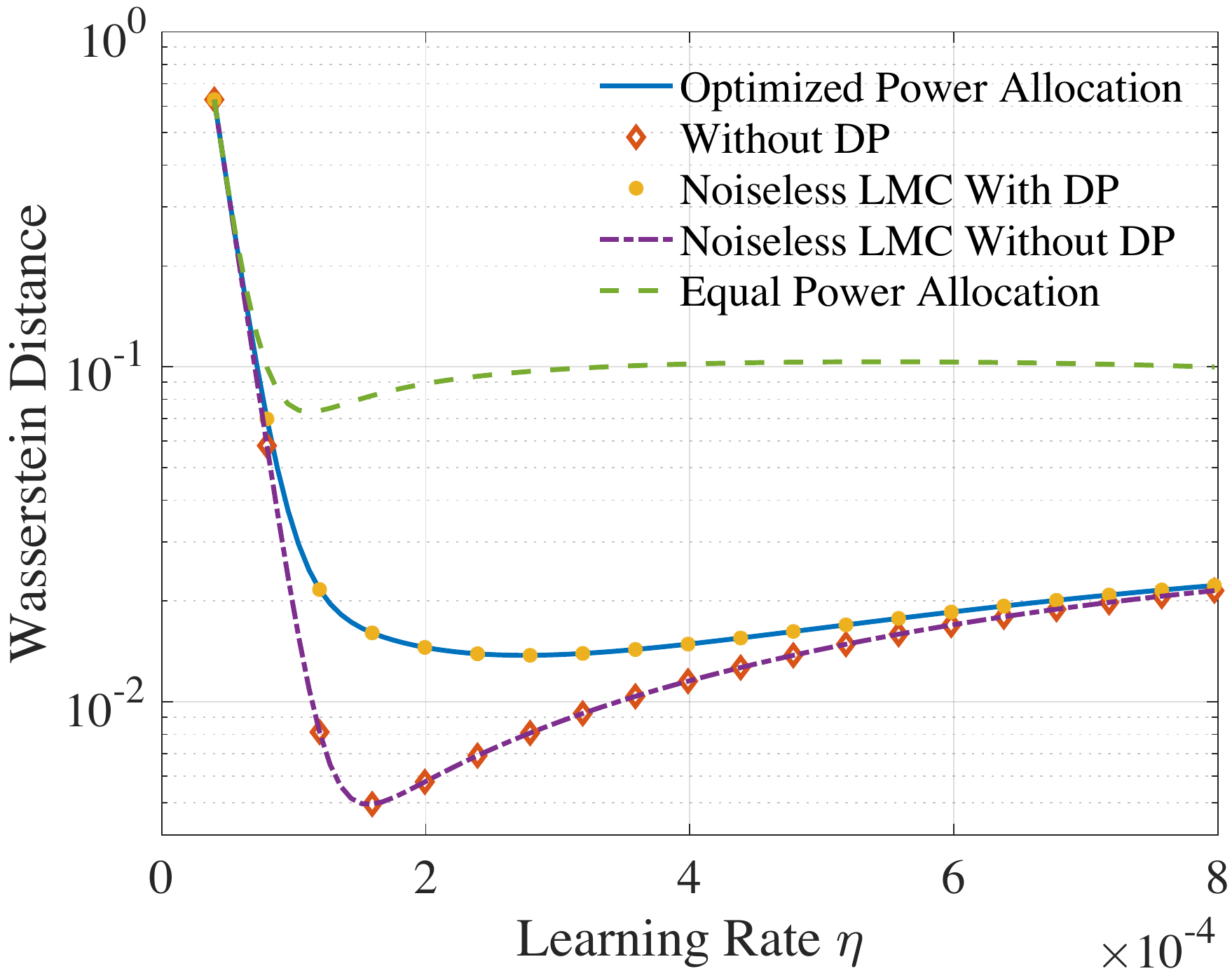}\label{Fig: etaSNR30}}
\vspace{-3mm} \caption{Bound \eqref{eq: result convergence} on the 2-Wasserstein distance  versus learning rate $\eta$ for WFLMC and noiseless LMC schemes ($\epsilon=8$, $\delta=0.01$,  $S_b=50$, $S_u=1$).}
\label{Fig: eta}
\vspace{-6mm}
\end{figure}

\noindent $\bullet$ {\bf Impact of the learning rate:}  We now further elaborate on the impact of the learning rate $\eta$ in Fig.~\ref{Fig: eta} by setting $\SNR=18$ dB and $\SNR=30$ dB. First, with $\eta\leq 0.5\times 10 ^{-4}$, by Fig.~\ref{Fig: partition}, we are in the LMC-limited regime where all the schemes have same performance. When $\eta> 0.5\times 10 ^{-4}$, the performance of optimized power allocation is shown to be power-limited for $\SNR=18$ dB in Fig. \ref{Fig: etaSNR18}, as it is identical to that without the DP constraint. In contrast, when $\SNR=30$ dB, as seen in Fig. \ref{Fig: etaSNR30}, the performance of  optimized power allocation is  limited by DP for $\eta> 0.5\times 10 ^{-4}$. As for the static power allocation, the performance is always limited by DP for $\eta> 0.5\times 10 ^{-4}$, and increasing the SNR is not helpful in reducing the approximation error.


\begin{figure}[t]
\centering
\begin{minipage}{0.48\textwidth}
\centering
\includegraphics[width=7cm]{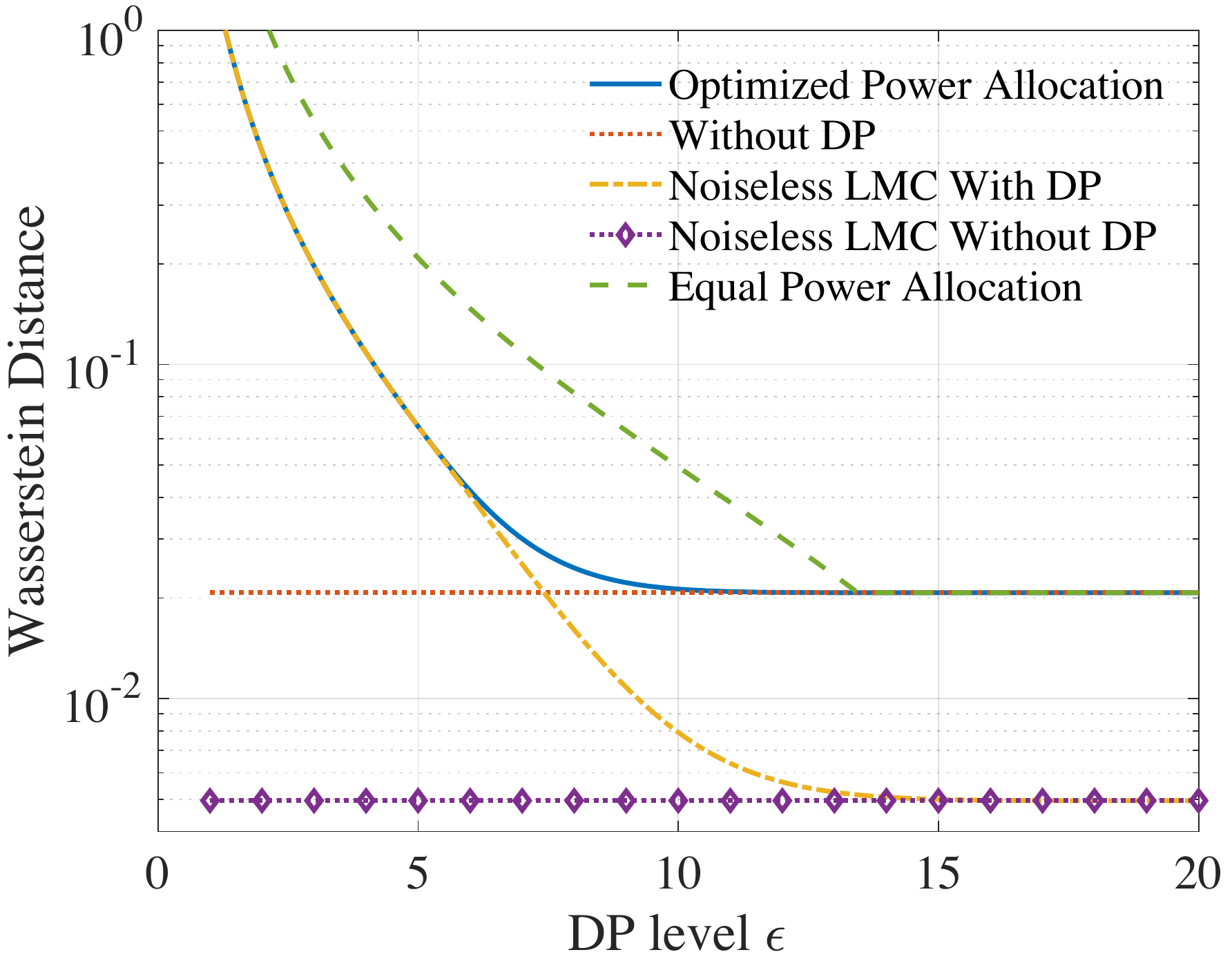}
\vspace{-6mm}
\caption{Bound  \eqref{eq: result convergence} on the 2-Wasserstein distance versus DP level $\epsilon$ for WFLMC and noiseless LMC schemes ($\eta=0.4/(\mu+L)$, $\delta=0.01$, $\SNR=20$ dB, $S_b=50$, $S_u=1$).}\label{Fig: dp}
\vspace{-2mm}
\end{minipage}
  \hspace{1mm}
\begin{minipage}{0.48\textwidth}
\centering    
\includegraphics[width=7cm]{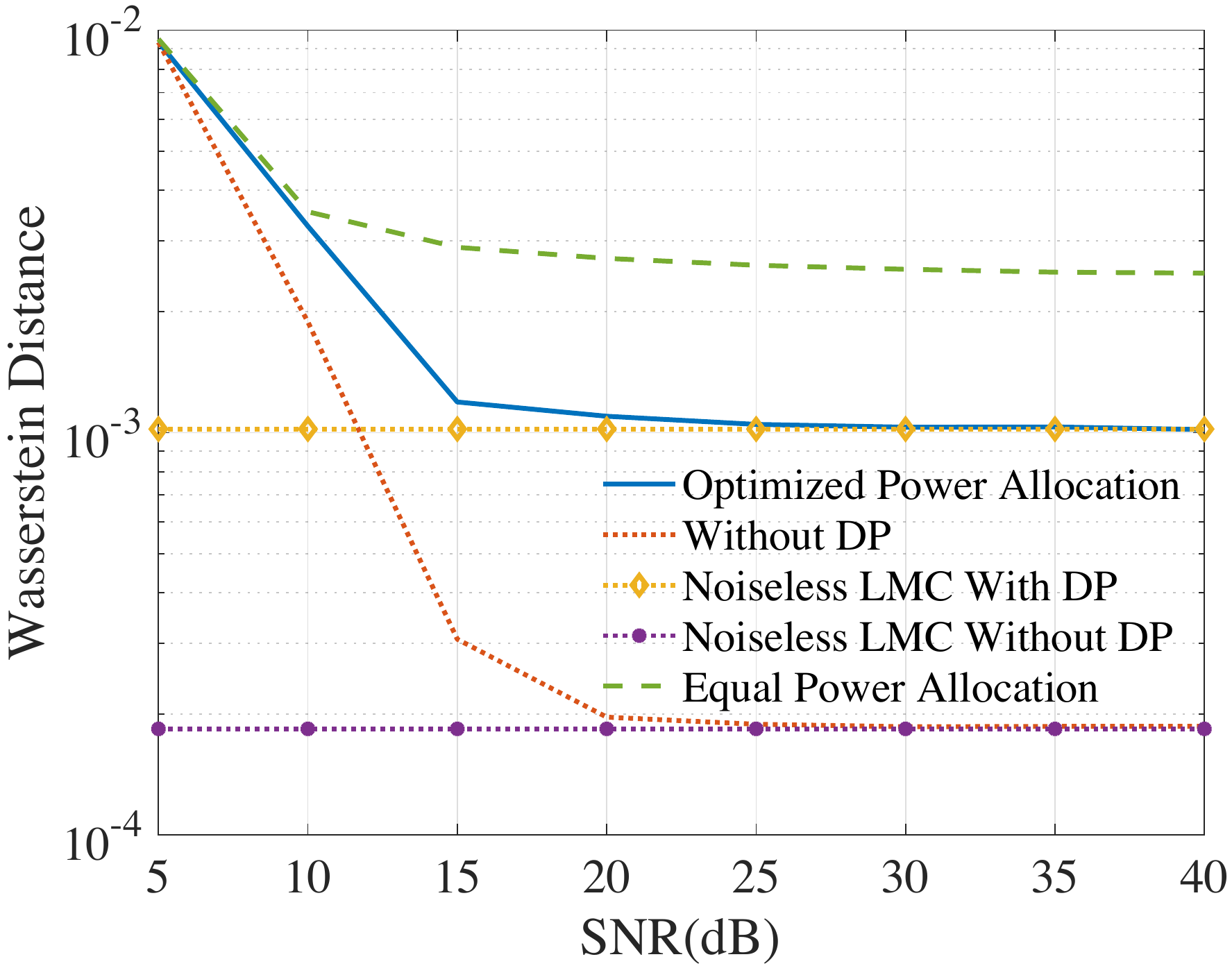}
\vspace{-6mm}
\caption{Evaluation of the 2-Wasserstein distance \eqref{eq: gauss ws eval} versus  SNR for WFLMC and noiseless LMC schemes ($\epsilon=15$, $\delta=0.01$, $S=100$, $S_u=50$).} \label{Fig: MultiSNR}
\vspace{-2mm}
\end{minipage}
\vspace{-4mm}
\end{figure}

\noindent $\bullet$ {\bf Impact of the privacy level:} Fig.~\ref{Fig: dp} plots the bound  \eqref{eq: result convergence} on the  2-Wasserstein distance as a function of the privacy level $\epsilon$. The performance of WFLMC is limited by DP until $\epsilon=13.5$, after which it becomes limited by the transmitted power and the DP constraint does not cause a performance loss. In the DP-limited regime, the proposed optimized power allocation outperforms static power allocation.  Furthermore, under a stricter DP requirement, i.e., with $\epsilon \leq 6$, the performance of optimized power allocation is equivalent to that under ideal communication, and the existence of channel noise does not impair performance.

\subsection{Multi-Sample Case}
We obtain the power control strategy in the multi-sample case by  solving the convex problem \eqref{ms wt OPT1} under the suboptimal  scheduling policy in \eqref{eq: min error g}. The optimized parameters are then used to implement WLMC. To evaluate the 2-Wasserstein distance, we numerically estimate the mean and covariance matrices of the samples over 100 experiments, and compare with the true posterior via \eqref{eq: gauss ws eval}.  The results present as the worst performance after the burn-in period, as given in the objective \eqref{opt obj}. Unless stated otherwise, the channels $h_k^{[s]}$ are randomly generated following $\mathcal{CN}(0, 0.01)$ for all $k$; the number of communication round is set to $S=100$; DP parameters are set to $\delta=0.01$ and $\epsilon=15$.

\noindent $\bullet$ {\bf Impact of the SNR:}
We first study the impact of the $\SNR=P/(mN_0)$ with number of samples $S_u=50$. The Wasserstein distances of all the wirelesses LMC schemes are seen to decrease with the SNR until approaching the DP-limited regime for optimized power allocation, and the LMC-limited regime for the scheme without DP constraint. The results emphasize the importance of optimizing power allocation in high SNR regime where the static power allocation is seen to have a significant performance degradation in DP-limited regime. Furthermore, they validate the insights obtained for the analysis under the simplified assumptions considered in~Theorem \ref{thm: optimal solution single sample}.

\noindent $\bullet$ {\bf Impact of the Burn-in Period:} In Fig. \ref{Fig: MultiBurnIn}, we investigate the impact of the number of communication periods, $S_b$, allocated for the burn-in period where we fix as $S=100$ the total number of communication rounds. In this experiment, the SNR is set to $30$ dB. Fig.  \ref{Fig: MultiBurnIn} shows that increasing the communication rounds in the burn-in period helps improve the quality of the samples produced after the burn-in period, as the Wasserstein distances of all the schemes are seen to decrease with the burn-in period. We also note that the enhanced sample quality costs at the cost of sample quantity. The performance of LMC is limited by the bias in discretization error that is determined by the learning rate and can not be diminished by increasing the duration of the burn-in period. Overall, the results emphasize that sample quality can be enhanced by optimizing the power allocation strategy. 

\begin{figure}[t]
\centering
\begin{minipage}{0.48\textwidth}
\centering    
\includegraphics[width=7cm]{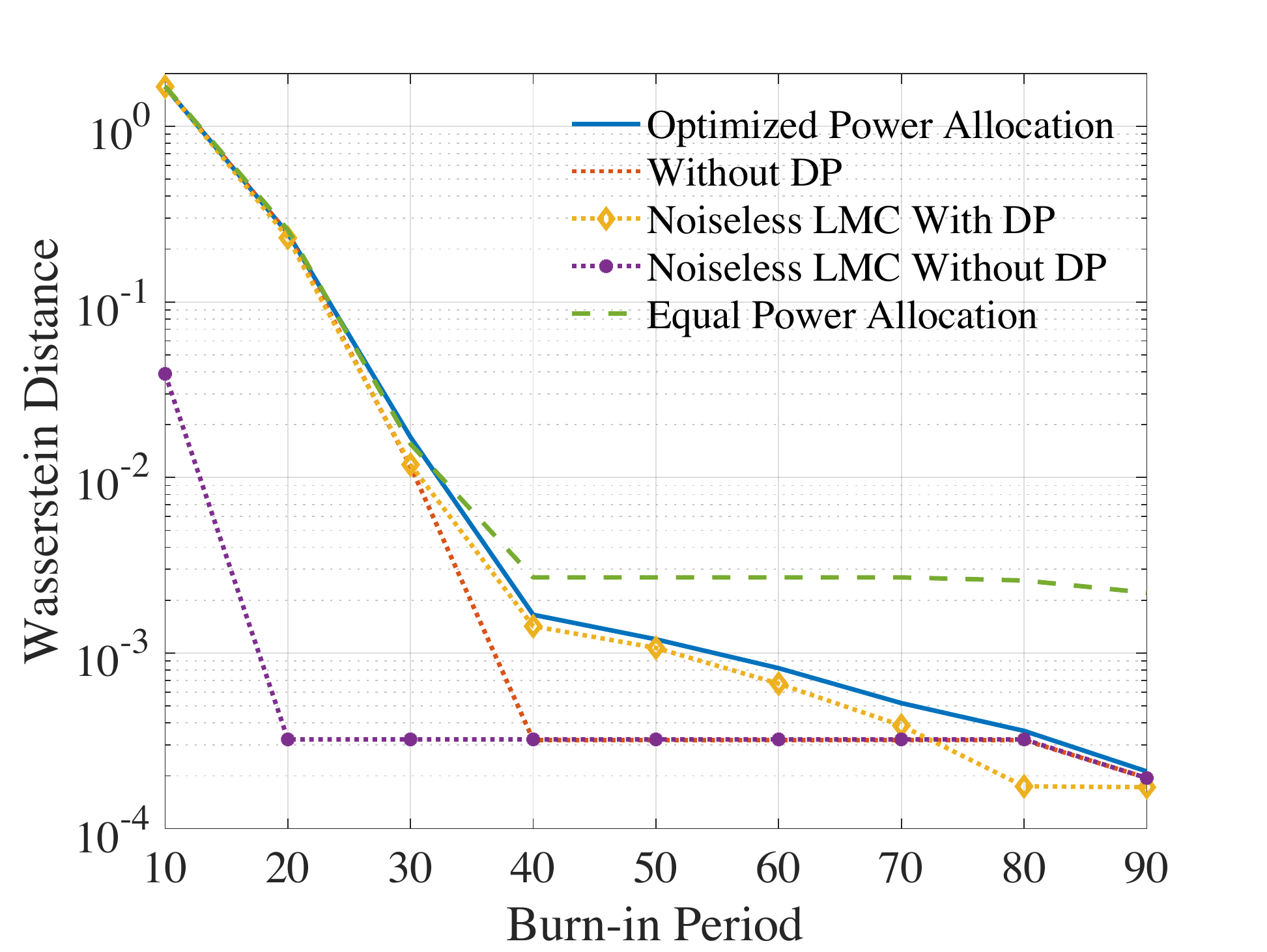}\vspace{-6mm}
\caption{Evaluation of the 2-Wasserstein distance \eqref{eq: gauss ws eval} versus the number of communication rounds in burn-in period $S_b$ for WFLMC and noiseless LMC schemes ($\epsilon=15$, $\delta=0.01$, $S=100$, $\SNR=30$ dB).}
\label{Fig: MultiBurnIn}
\vspace{-6mm}
\end{minipage}
  \hspace{1mm}
\begin{minipage}{0.48\textwidth}
\centering
\includegraphics[width=7cm]{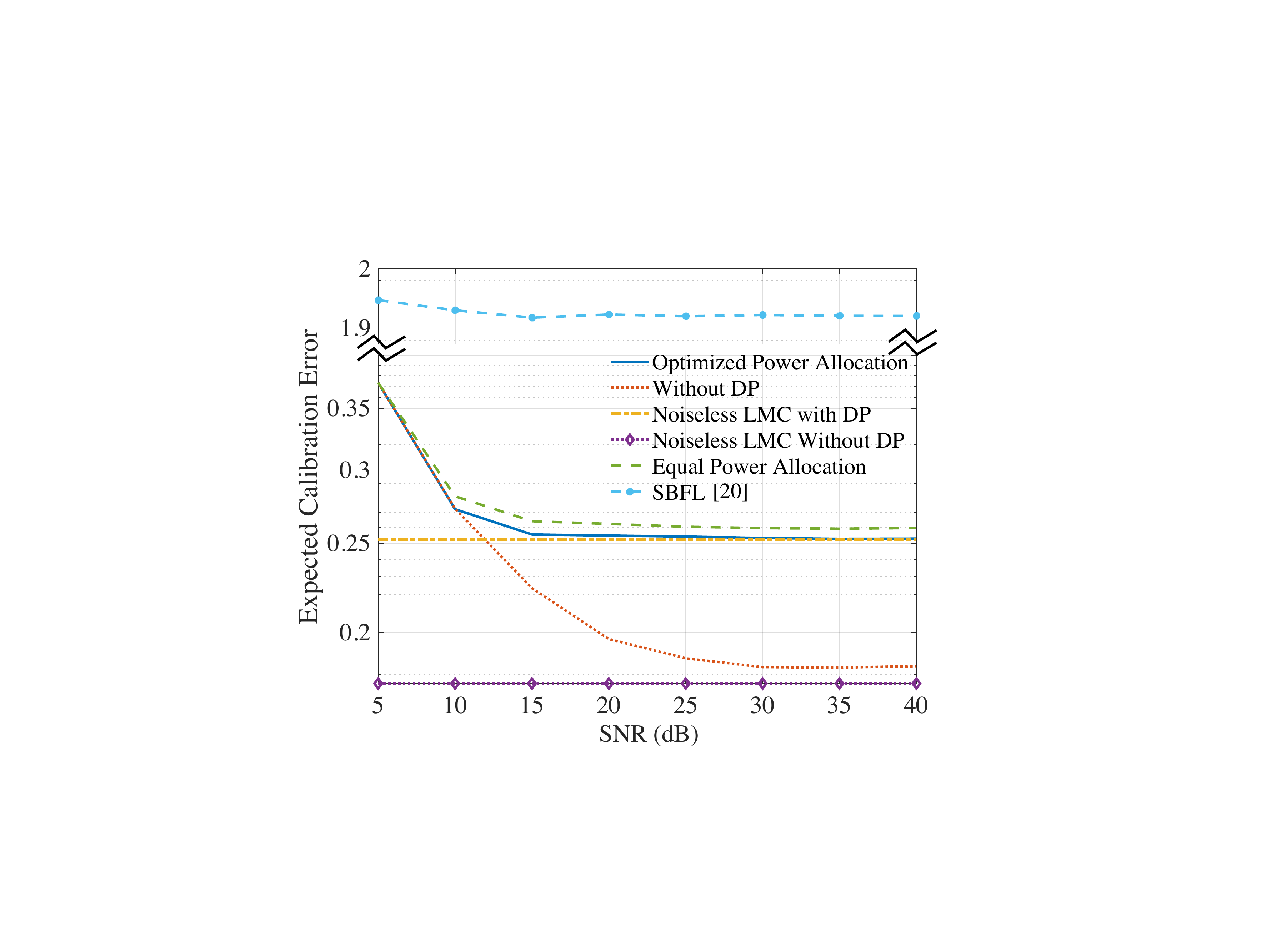}\vspace{-6mm}
\caption{Evaluation of expected calibration error versus  SNR for WFLMC, noiseless LMC schemes, and SBFL in \cite{lee2020bayesian} ($\epsilon=50$, $\delta=0.1$, $S=500$, $S_u=50$). \label{Fig: MNIST_SNR}}
\vspace{-6mm}
\end{minipage}
\end{figure}

{\color{black}
 \subsection{MNIST Data Set}

We now consider the problem of multinomial logistic regression on the MNIST data set, which is comprised of $C=10$ classes representing handwritten digits. We use $N=12,000$ data points, and $N/C$ data points for each class are randomly selected from the MINST training set. The original data, with dimension $784$,  is pre-processed by  projecting the input images into a subspace of lower dimension $30$ via principal component analysis (PCA). This is motivated by the fact that the MNIST images are known to be approximately supported on a manifold with intrinsic dimension lower than 30 \cite{hein2005intrinsic}.

For the learning model, the likelihood is considered as 
\begin{align}
 p(v_n | \{ {\bm \theta}_c\}_{c=1}^C, \bu_n)= \prod_{c=1}^C \bigg[   \frac{\exp({\bm \theta}_c ^{\sf T} \bu_n)}{ \sum_{c'=1}^C \exp({\bm \theta}_{c'}^{\sf T} \bu_n) }\bigg]^{I_{n,c}}, \tag{60}
\end{align}
where we have $I_{n,c}=  {\bf 1} [v_n= c] $. The prior of the model parameters ${\bm \theta}=\{\bm \theta_c\}_{c=1}^C$ is a produced Gaussian distribution $\cN (0, \bI_m)$ for $m=300$. To set the parameters $\mu$ and $L$, we leverage the inequality for the Hessian matrix  $  \bI_m   \preceq  \nabla^2 f ({\bm \theta}) \preceq \frac{1}{2} \l(\bI_c- \frac{1}{C}{\bf 1}_{c \times c}\r) \otimes  \bU \bU^{\sf T} +\bI_m $, where ${\bf 1}_{c\times c}$ is a $c\times c$ matrix with all elements $1$, and $\otimes$ represents Kronecker product. Specifically, we compute the minimum and maximum eigenvalues of the lower bound and upper bound respectively to set $\mu$ and $L$. The channels $h_k^{[s]}$ are randomly generated following the distribution $\mathcal{CN}(0, 10^{-4})$ for all $k$; the number of communication round is set to $S=600$; the number of communication burn-in periods is set to $S_b=550$; the clipping threshold in \eqref{eq: BV} is set to $\ell=300$; and DP parameters are set to $\delta=0.1$ and $\epsilon=50$.  The results are averaged over 50 experiments. 

For tractability, the optimization problem aims at minimizing the averaged Wasserstein distances across the samples after the burn-in period. The problem is convex and can be solved by the method in Sec \ref{sec: single sample case}.  We use the test set of MNIST to evaluate the expected calibration error \cite{guo2017calibration}. Furthermore, we implemented scalable Bayesian federated learning (SBFL) in \cite{lee2020bayesian} as another benchmark. We recall that SBFL is a frequentist learning protocol that aggregates the local gradients at the server via a point estimate that incorporates prior knowledge on the gradients. Under the assumption of orthogonal multiple access in \cite{lee2020bayesian},  the iteration time is set as $S/K$ for a fair comparison.

Fig. \ref{Fig: MNIST_SNR} plots the expected calibration error versus SNR with the number of samples $S_u=50$. First, the figure verifies the advantage of Bayesian learning, as all WFLMC schemes are seen to significantly outperform SBFL. Furthermore, by increasing the number of communication rounds, optimized power allocation can outperform the equal power allocation scheme.
}

\section{Conclusions}\label{sec: conclusions}

In this paper,  we have proposed a novel Bayesian federated learning (FL) protocol that implements Langevin Monte Carlo (LMC) via uncoded wireless transmission from devices to edge server. The learning protocol is enabled by over-the-air computing, which is akin to  previous works on wireless (frequentist) FL \cite{zhu2019broadband,liu2020privacy}, as well as by the novel idea of repurposing channel noise for MC sampling in Bayesian FL. The goal is to obtain samples at the edge server that are approximately drawn according to the global posterior distribution. 
The key idea of this work is to leverage the channel noise for both MC sampling and  privacy preservation. 
Under simplified assumptions, 
the analysis has revealed that system operating different regimes limited by LMC, DP or power constraints, depending on the values of learning rate, privacy parameters, and transmitted power. As for the general case, the problem of optimizing power allocation was proved to be convex. Simulation results have validated the analysis, yielding insights into conditions under which channel noise is not harmful to the system performance. 

As an extension of the current work,  it is interesting to optimize the learning rate schedule, which may not only benefit learning performance but also enhance privacy. 
Another direction is to study digital implementations of  wireless federated LMC as in \cite{zhu2020one}.  The current work can also be further generalized to multi-hop device-to-device (D2D) network topologies.  {\color{black}The analysis could also be extended to the scenario of non-ideal downlink communication as in \cite{wei2021federated, amiri2021convergence}.}

\appendix 

\subsection{Proof of Lemma~\ref{lemma: gradient error}}\label{proof: gradient error}
According to \eqref{eq: LMC wireless gd noise} and the definition of the gradient error \eqref{eq: grad error}, we have 
\begin{align}
\E\l[\| N_g^{[s]} \|^2 \r]&= \E \bigg[\bigg\| \frac{K}{K_a^{[s]}} \sum_{k\in \cK_a^{[s]} } \nabla f_k\big({\bm \theta}^{[s-1]}\big) + \sqrt{\widetilde{\beta}^{(s)}}\Delta^{[s]}  - \sum_{k=1 }^K \nabla f_k\big({\bm \theta}^{[s-1]}\big) \bigg\|^2\bigg]\nn\\
&=\E \bigg[\bigg\| \frac{K}{K_a^{[s]}} \sum_{k\in \cK_a^{[s]} } \nabla f_k\big({\bm \theta}^{[s-1]}\big)- \sum_{k=1 }^K \nabla f_k\big({\bm \theta}^{[s-1]}\big)    \bigg\|^2\bigg]+  \widetilde{\beta}^{(s)} \nn\\
&\overset{(a)}\leq \E\bigg[ \bigg(  \frac{K-K_a^{[s]} }{K_a^{[s]}} \sum_{k\in \cK_a^{[s]} } \big\|\nabla f_k\big({\bm \theta}^{[s-1]}\big)\big \| + \sum_{k\notin \cK_a^{[s]} } \big\|\nabla f_k\big({\bm \theta}^{[s-1]}\big)  \big\|\bigg)^2 \bigg]+  \widetilde{\beta}^{(s)}  \nn\\
&\overset{(b)}\leq 4\ell^2 \big(K-K_a^{[s]}\big)^2+\widetilde{\beta}^{(s)},  \nn
\end{align}
where $(a)$ is obtained by applying triangle inequality and $(b)$ following from Assumption \ref{assumption: BLL}. 

\subsection{Proof of Lemma~\ref{lemma: discr error}}\label{proof: discr error}
The vector $\bar{\bm \theta}^{(t)}$ is the continuous-time Langevin diffusion process in \eqref{eq: LD} and its distribution is assumed to be invariant, so that we have $N_d^{[s]}=N_d^{[1]}$ for all $s$. 
{\color{black}The proof follows from \cite[Lemma 3]{dalalyan2017further} with the caveat that we bound of $\E [\| N_d^{[1]} \|^2]$ in lieu of $\sqrt{\E [\| N_d^{[1]} \|^2]} $.  This is done as:
\begin{align}
&\E\big[\|N_d^{[s]}\|^2\big] = \E\big[\|N_d^{[1]}\|^2\big] =\E \l[\Big\|\int_{0}^{\eta}  \nabla f(\bar{\bm \theta}^{(t)})  - \nabla f(\bar{\bm \theta}^{(0)})   \mathrm{d} t  \Big\| ^2\r] \nn\\ 
&\overset{(a)}\leq \int_{0}^{\eta}  \E \l[ \Big\|  \nabla f(\bar{\bm \theta}^{(t)})  - \nabla f(\bar{\bm \theta}^{(0)}) \Big\| ^2  \r]     \mathrm{d} t \cdot \int_{0}^{\eta}  1 \mathrm{d} t   \overset{(b)}= \eta L^2   \int_{0}^{\eta}   \E \l[  \Big\|  -\int_{0}^{t}  \nabla f(\bar{\bm \theta}^{(s)})  \mathrm{d} s + \sqrt{2t} {\bm \xi}^{(1)}  \Big\| ^2 \r] \mathrm{d} t     \nn\\
&= \eta L^2   \int_{0}^{\eta}  \l( \E \l[  \Big\|  -\int_{0}^{t}  \nabla f(\bar{\bm \theta}^{(s)})  \mathrm{d} s\Big\| ^2  \r]    +2 (m \cdot t) \r)\mathrm{d} t     \nn\\ 
&\overset{(c)}\leq   \eta L^2   \int_{0}^{\eta}   t \int_{0}^{t}   \E \l[ \big\|  \nabla f(\bar{\bm \theta}^{(s)}) \big\| ^2 \r]   \mathrm{d} s  \mathrm{d} t      +\eta^3 L^2m   \overset{(d)} \leq \frac{\eta^4L^3 m}{3}+ \eta^3 L^2m,  \label{eq: nd}
\end{align}
where $(a)$ and $(c)$ are consequences of Cauchy-Schwarz inequality and of the interchanging the order of expectation and integral; $(b)$ is obtained by Assumption \ref{assumption: Lipschitz} and  \eqref{eq: diff LD}; $(d)$ is derived using \cite[Lemma 2]{dalalyan2017further}. }

\subsection{Proof of Proposition \ref{prop: convergence}}\label{proof: convergence}
 For any $\tau>0$, the error \eqref{eq: w2 mid} is bounded as 
  \begin{align}\label{eq: w2 mid ub}
{\E\l[\l\| {\bm \theta}^{[s]}-  \bar{\bm \theta}^{(s\eta)}\r\|^2\r]}&\leq (1+\tau) \E\bigg[\Big\| {\bm \theta}^{[s-1]}- \bar{\bm \theta}^{( (s-1)\eta)} - \eta \big[\nabla f({\bm \theta}^{[s-1]} ) - \nabla f(  \bar{\bm \theta}^{((s-1)\eta)} )\big] \Big\|^2 \bigg] \nn\\
&\qquad \qquad\qquad\qquad\qquad\qquad\qquad\quad+ (1+\tau^{-1}) \E[\|\eta N_g^{[s]}+N_d^{[s]}\|^2]  \nn\\ 
&\leq (1+\tau) \E\bigg[\Big\| {\bm \theta}^{[s-1]}- \bar{\bm \theta}^{( (s-1)\eta)} - \eta \big[\nabla f({\bm \theta}^{[s-1]} ) - \nabla f(  \bar{\bm \theta}^{((s-1)\eta)} )\big] \Big\|^2 \bigg] \nn\\
&\qquad \qquad\qquad\qquad\qquad + (1+\tau^{-1}) \l(2\eta^2\E[\| N_g^{[s]}\|^2]+2\E[\|N_d^{[s]}\|^2] \r),
\end{align} 
where the value of $\tau$ controls the convergence rate as detailed later. 
Furthermore, the first term is bounded by 
\begin{align}
\E\bigg[\Big\| {\bm \theta}^{[s-1]}- \bar{\bm \theta}^{( (s-1)\eta)} - \eta \big[\nabla f({\bm \theta}^{[s-1]} ) - \nabla f(  \bar{\bm \theta}^{((s-1)\eta)} )\big] \Big\|^2 \bigg]\leq \gamma^2 \E\l[\l\| {\bm \theta}^{[s-1]}-  \bar{\bm \theta}^{((s-1)\eta)}\r\|^2\r], \!\! \label{eq: ft bound}
\end{align}
where $\gamma=1-\eta \mu$  if $0<\eta\leq 2/(\mu+L)$; and $\gamma= \eta L-1 $  if $2/(\mu+L)\leq \eta \leq 2/L$. The proof starts with \eqref{eq: smooth and sc 2} through \cite[Lemma 1]{dalalyan2017further}. Specifically, plugging Lemmas \ref{lemma: gradient error} and \ref{lemma: discr error} and \eqref{eq: ft bound} into \eqref{eq: w2 mid ub}, and setting  $\tau=(\frac{1+\gamma}{2\gamma})^2-1$ yield 
\begin{align}\label{eq: rec wd}
&{\E\l[\l\| {\bm \theta}^{[s]}-  \bar{\bm \theta}^{(s\eta)}\r\|^2\r]} \leq \Big(\frac{1+\gamma}{2}\Big)^2{\E\l[\l\| {\bm \theta}^{[s-1]}-  \bar{\bm \theta}^{((s-1)\eta))}\r\|^2\r]}+ \frac{2(1+\gamma)^2}{(1-\gamma)(1+3\gamma)} \nn\\
 &  \qquad  \qquad \qquad \times \bigg[ \frac{\eta^4L^3 m}{3}+ \eta^3 L^2m+ 4\eta^2\ell^2\big(K-K_a^{[s]}\big)^2  + \max \bigg\{0, \frac{\eta^2N_0K^2}{(\alpha^{[s] } K_a^{[s]})^2} -{2}{\eta} \bigg\} \bigg] .
\end{align}
The desired result in Proposition \ref{prop: convergence} is obtained by using the inequality $\frac{1+\gamma}{1+3\gamma}<1$, applying \eqref{eq: rec wd} recursively, and the initial point ${\bm \theta}^{[0]}$ is set to satisfy $W_2\big(p^{[0]}(\bm \theta),p({\bm\theta}|\cD)\big)^2= {\E[\| {\bm \theta}^{[0]}-  \bar{\bm \theta}^{(0)}\|^2]}$.

\subsection{Proof of Theorem~\ref{thm: optimal solution single sample}} \label{proof: optimal solution single sample} 
Substituting  $(\alpha^{[s] } )^2$ with $a^{[s]}$, and $a^{[s]}\geq0$,  the optimization problem \eqref{os wt OPT1} is equivalent to the following convex program
\begin{subequations}\label{os wt OPT2}
 \begin{align}
{(\text {\bf Case 1 Opt. Equiv.})} \    \min_{{\{a^{[s]}\}}_{s=1}^{S_b+1}} &\quad  \sum_{s=1}^{S_b+1} \Big(\frac{1+\gamma}{2}\Big)^{-2s} \l( \frac{\eta^2N_0K^2}{a^{[s ]} K_a^2} -{2}{\eta}\r) \label{os wt obj2} \\
{\rm s.t.} \quad &
\sum_{s=1}^{S_b+1} a^{[s]}\leq  \frac{ N_0\mathcal{R}_{\sf dp}(\epsilon, \delta)}{2 \ell^2},  \label{os wt dp2}\\
&  a^{[s]} \leq \min_{k\in \cK_a}   \frac{P | h_k|^2}{\ell^2} , \quad \forall  k \in \cK_a, s=1,\cdots, S,   \label{os wt power constraint2} \\
& 0\leq a^{[s]}\leq \frac{\eta N_0 K^2}{2 K_a^2}, \quad \forall  k \in \cK_a, s=1,\cdots, S. 
 \end{align}
\end{subequations}
To solve this problem, the partial Lagrange function is defined as 
 \begin{align}
\mathcal{L}= \sum_{s=1}^{S_b+1} \Big(\frac{1+\gamma}{2}\Big)^{-2s} \frac{\eta^2N_0K^2}{a^{[s ]} K_a^2}+  \lambda \l(\sum_{s=1}^{S_b+1}  {a^{[s]}}-\frac{ N_0\mathcal{R}_{\sf dp}(\epsilon, \delta)}{2 \ell^2}  \r)
 - \sum_{s=1}^{S_b+1}  \psi^{[s]}a^{[s]}  \nn\\
 +\sum_{s=1}^{S_b+1} \varepsilon^{[s]} \bigg(a^{[s]}-\min\Big\{  \min_{k\in \cK_a }   \frac{P |h_k|^2}{\ell^2},\frac{\eta N_0K^2}{2K_a^2} \Big\}\bigg),
 \end{align}
where $\lambda\geq0$, $\psi^{[s]}\geq0$,  and $ \varepsilon^{[s]}\geq0$ are the Lagrange multipliers associated respectively with the DP constraint,  non-negative parameter constraints, joint transmit power and non-negative noise variance constraints.    Then, applying the KKT conditions leads to the following necessary and sufficient conditions for optimality
  \begin{subequations}\label{eq: KKT case1}
 \begin{align}
&\frac{\partial \mathcal{L}}{\partial( a_{\sf opt}^{[s]})}=-\Big(\frac{1+\gamma}{2}\Big)^{-2s} \frac{\eta^2N_0K^2}{(a^{[s] }_{\sf opt} K_a)^2  }+ \lambda_{\sf opt} -\psi_{\sf opt}^{[s]}+\varepsilon_{\sf opt}^{[s]} =0 , \  \forall s ,   \label{case1: KKT lanrange at} \\
&\lambda \l(\sum_{s=1}^{S_b+1}  {a_{\sf opt}^{[s]}}-\frac{ N_0\mathcal{R}_{\sf dp}(\epsilon, \delta)}{2 \ell^2}  \r)=0, \label{case1: KKT DP}\\
&  \varepsilon_{\sf opt}^{[s]} \bigg(a_{\sf opt}^{[s]}-\min\Big\{  \min_{k\in \cK_a }   \frac{P |h_k|^2}{\ell^2},\frac{\eta N_0K^2}{2K_a^2} \Big\}\bigg)=0,  \  \forall s  ,  \label{case1: KKT power constraint}\\
&  -\psi_{\sf opt}^{[s]}a_{\sf opt}^{[s]} =0,\  \forall s ,     \label{case1: KKT non zero} \\
& \sum_{s=1}^{S_b+1}  {a_{\sf opt}^{[s]}}-\frac{ N_0\mathcal{R}_{\sf dp}(\epsilon, \delta)}{2 \ell^2} \leq 0 ,  \label{case1: KKT ineq DP}\\
&a_{\sf opt}^{[s]}-\min\Big\{  \min_{k\in \cK_a }   \frac{P |h_k|^2}{\ell^2},\frac{\eta N_0K^2}{2K_a^2} \Big\}  \leq 0, \ \forall s,   \label{case1: KKT ineq power} \\
& -a^{(t)}_{\sf opt} \leq 0.  \label{case1: KKT ineq non zero} 
 \end{align} 
 \end{subequations}
 Combing \eqref{case1: KKT lanrange at} and \eqref{case1: KKT non zero} we have 
\begin{align}
(a^{[s] }_{\sf opt})^2 = \Big(\frac{1+\gamma}{2}\Big)^{-2s} \frac{\eta^2N_0 K^2}{K_a^2\big(\lambda_{\sf opt} +\varepsilon_{\sf opt}^{[s]}\big)}
\end{align}
On the other hand, \eqref{case1: KKT power constraint} indicates, that if $\varepsilon_{\sf opt}^{[s]}>0$, then we have 
\begin{align}\label{case1: opta without DP}
a^{[s] }_{\sf opt} =\min\Big\{  \min_{k\in \cK_a^{[s]}}   \frac{P  |h_k|^2}{\ell^2},\frac{\eta N_0K^2}{2 K_a ^2} \Big\}.  
\end{align}
Furthermore, according to \eqref{case1: KKT DP},  if  condition 
\begin{align}\label{eq: equality DP}
\sum_{s=1}^{S_b+1}  {a_{\sf opt}^{[s]}}<\frac{ N_0\mathcal{R}_{\sf dp}(\epsilon, \delta)}{2 \ell^2}  
\end{align}
holds, we have $\lambda_{\sf opt}=0$, and thus $\varepsilon_{\sf opt}^{[s]}>0$ for all $s$ with the optimal value of $a^{[s] }_{\sf opt} $ given in  \eqref{case1: opta without DP}. 
Otherwise we have the following equality 
\begin{align}\label{eq: equality lambda}
\sum_{s=1}^{S_b+1}  {a_{\sf opt}^{[s]}}-\frac{ N_0\mathcal{R}_{\sf dp}(\epsilon, \delta)}{2 \ell^2} = 0, 
\end{align} 
and the optimal solution is 
\begin{align}\label{case1: opta with DP}
a^{[s] }_{\sf opt} =\min\bigg\{ \Big(\frac{1+\gamma}{2}\Big)^{-s}\frac{\eta K N_0^{1/2}}{K_a   \lambda_{\sf opt}^{1/2}}, \min_{k\in \cK_a}   \frac{P  |h_k|^2}{\ell^2},\frac{\eta N_0K^2}{2 K_a ^2} \bigg\}.
\end{align}
The solution of $ \lambda_{\sf opt} $ is obtained by solving \eqref{eq: equality lambda}.  Reverting to the original variables via $\alpha^{[s]}=\sqrt{a^{[s]}}$ and separating the solution \eqref{case1: opta without DP}  into $a_{\sf opt}^{[s]}=\min_{k\in\cK_a} P |h_k|^2/\ell^2$ and $a_{\sf opt}^{[s]}=\eta N_0 K^2/ (2K_a^2)$ yield desired results in the theorem. 
 

\end{document}

%% file: header.tex
\newtheorem{theorem}{Theorem}
\newtheorem{acknowledgement}[theorem]{Acknowledgement}
\newtheorem{axiom}[theorem]{Axiom}
\newtheorem{case}[theorem]{Case}
\newtheorem{claim}[theorem]{Claim}
\newtheorem{conclusion}[theorem]{Conclusion}
\newtheorem{condition}[theorem]{Condition}
\newtheorem{conjecture}[theorem]{Conjecture}
\newtheorem{criterion}[theorem]{Criterion}
\newtheorem{definition}{Definition}
\newtheorem{exercise}[theorem]{Exercise}
\newtheorem{lemma}{Lemma}
\newtheorem{corollary}{Corollary}
\newtheorem{notation}[theorem]{Notation}
\newtheorem{problem}[theorem]{Problem}
\newtheorem{proposition}{Proposition}
\newtheorem{scheme}{Scheme}   
\newtheorem{solution}[theorem]{Solution}
\newtheorem{summary}[theorem]{Summary}
\newtheorem{assumption}{Assumption}
\newtheorem{example}{\bf Example}
\newtheorem{remark}{\bf Remark}

\def\qed{$\Box$}
\def\QED{\mbox{\phantom{m}}\nolinebreak\hfill$\,\Box$}
\def\proof{\noindent{\emph{Proof:} }}
\def\poof{\noindent{\emph{Sketch of Proof:} }}
\def
\endproof{\hspace*{\fill}~\qed
\par
\endtrivlist\unskip}
\def\endproof{\hspace*{\fill}~\qed\par\endtrivlist\vskip3pt}

\def\E{\mathsf{E}}
\def\eps{\varepsilon}
\def\phi{\varphi}
\def\Lsp{{\boldsymbol L}}
\def\Bsp{{\boldsymbol B}}
\def\lsp{{\boldsymbol\ell}}
\def\Ltsp{{\Lsp^2}}
\def\Lpsp{{\Lsp^p}}
\def\Linsp{{\Lsp^{\infty}}}
\def\LtR{{\Lsp^2(\Rst)}}
\def\ltZ{{\lsp^2(\Zst)}}
\def\ltsp{{\lsp^2}}
\def\ltZt{{\lsp^2(\Zst^{2})}}
\def\ninN{{n{\in}\Nst}}
\def\oh{{\frac{1}{2}}}
\def\grass{{\cal G}}
\def\ord{{\cal O}}
\def\dist{{d_G}}
\def\conj#1{{\overline#1}}
\def\ntoinf{{n \rightarrow \infty }}
\def\toinf{{\rightarrow \infty }}
\def\tozero{{\rightarrow 0 }}
\def\trace{{\operatorname{Tr}}}
\def\ord{{\cal O}}
\def\UU{{\cal U}}
\def\rank{{\operatorname{rank}}}
\def\acos{{\operatorname{acos}}}

\def\SINR{\mathsf{SINR}}
\def\SNR{\mathsf{SNR}}
\def\SIR{\mathsf{SIR}}
\def\tSIR{\widetilde{\mathsf{SIR}}}
\def\Ei{\mathsf{Ei}}
\def\l{\left}
\def\r{\right}
\def\({\left(}
\def\){\right)}
\def\lb{\left\{}
\def\rb{\right\}}

\setcounter{page}{1}

\newcommand{\eref}[1]{(\ref{#1})}
\newcommand{\fig}[1]{Fig.\ \ref{#1}}

\def\bydef{:=}
\def\ba{{\mathbf{a}}}
\def\bb{{\mathbf{b}}}
\def\bc{{\mathbf{c}}}
\def\bd{{\mathbf{d}}}
\def\bee{{\mathbf{e}}}
\def\bff{{\mathbf{f}}}
\def\bg{{\mathbf{g}}}
\def\bh{{\mathbf{h}}}
\def\bi{{\mathbf{i}}}
\def\bj{{\mathbf{j}}}
\def\bk{{\mathbf{k}}}
\def\bl{{\mathbf{l}}}
\def\bn{{\mathbf{n}}}
\def\bo{{\mathbf{o}}}
\def\bp{{\mathbf{p}}}
\def\bq{{\mathbf{q}}}
\def\br{{\mathbf{r}}}
\def\bs{{\mathbf{s}}}
\def\bt{{\mathbf{t}}}
\def\bu{{\mathbf{u}}}
\def\bv{{\mathbf{v}}}
\def\bw{{\mathbf{w}}}
\def\bx{{\mathbf{x}}}
\def\by{{\mathbf{y}}}
\def\bz{{\mathbf{z}}}
\def\b0{{\mathbf{0}}}

\def\bA{{\mathbf{A}}}
\def\bB{{\mathbf{B}}}
\def\bC{{\mathbf{C}}}
\def\bD{{\mathbf{D}}}
\def\bE{{\mathbf{E}}}
\def\bF{{\mathbf{F}}}
\def\bG{{\mathbf{G}}}
\def\bH{{\mathbf{H}}}
\def\bI{{\mathbf{I}}}
\def\bJ{{\mathbf{J}}}
\def\bK{{\mathbf{K}}}
\def\bL{{\mathbf{L}}}
\def\bM{{\mathbf{M}}}
\def\bN{{\mathbf{N}}}
\def\bO{{\mathbf{O}}}
\def\bP{{\mathbf{P}}}
\def\bQ{{\mathbf{Q}}}
\def\bR{{\mathbf{R}}}
\def\bS{{\mathbf{S}}}
\def\bT{{\mathbf{T}}}
\def\bU{{\mathbf{U}}}
\def\bV{{\mathbf{V}}}
\def\bW{{\mathbf{W}}}
\def\bX{{\mathbf{X}}}
\def\bY{{\mathbf{Y}}}
\def\bZ{{\mathbf{Z}}}

\def\mA{{\mathbb{A}}}
\def\mB{{\mathbb{B}}}
\def\mC{{\mathbb{C}}}
\def\mD{{\mathbb{D}}}
\def\mE{{\mathbb{E}}}
\def\mF{{\mathbb{F}}}
\def\mG{{\mathbb{G}}}
\def\mH{{\mathbb{H}}}
\def\mI{{\mathbb{I}}}
\def\mJ{{\mathbb{J}}}
\def\mK{{\mathbb{K}}}
\def\mL{{\mathbb{L}}}
\def\mM{{\mathbb{M}}}
\def\mN{{\mathbb{N}}}
\def\mO{{\mathbb{O}}}
\def\mP{{\mathbb{P}}}
\def\mQ{{\mathbb{Q}}}
\def\mR{{\mathbb{R}}}
\def\mS{{\mathbb{S}}}
\def\mT{{\mathbb{T}}}
\def\mU{{\mathbb{U}}}
\def\mV{{\mathbb{V}}}
\def\mW{{\mathbb{W}}}
\def\mX{{\mathbb{X}}}
\def\mY{{\mathbb{Y}}}
\def\mZ{{\mathbb{Z}}}

\def\cA{\mathcal{A}}
\def\cB{\mathcal{B}}
\def\cC{\mathcal{C}}
\def\cD{\mathcal{D}}
\def\cE{\mathcal{E}}
\def\cF{\mathcal{F}}
\def\cG{\mathcal{G}}
\def\cH{\mathcal{H}}
\def\cI{\mathcal{I}}
\def\cJ{\mathcal{J}}
\def\cK{\mathcal{K}}
\def\cL{\mathcal{L}}
\def\cM{\mathcal{M}}
\def\cN{\mathcal{N}}
\def\cO{\mathcal{O}}
\def\cP{\mathcal{P}}
\def\cQ{\mathcal{Q}}
\def\cR{\mathcal{R}}
\def\cS{\mathcal{S}}
\def\cT{\mathcal{T}}
\def\cU{\mathcal{U}}
\def\cV{\mathcal{V}}
\def\cW{\mathcal{W}}
\def\cX{\mathcal{X}}
\def\cY{\mathcal{Y}}
\def\cZ{\mathcal{Z}}
\def\cd{\mathcal{d}}
\def\Mt{M_{t}}
\def\Mr{M_{r}}
\def\O{\Omega_{M_{t}}}
\newcommand{\figref}[1]{{Fig.}~\ref{#1}}
\newcommand{\tabref}[1]{{Table}~\ref{#1}}

\newcommand{\var}{\mathsf{var}}
\newcommand{\fb}{\tx{fb}}
\newcommand{\nf}{\tx{nf}}
\newcommand{\BC}{\tx{(bc)}}
\newcommand{\MAC}{\tx{(mac)}}
\newcommand{\Pout}{p_{\mathsf{out}}}
\newcommand{\nnn}{\nn\\}
\newcommand{\FB}{\tx{FB}}
\newcommand{\TX}{\tx{TX}}
\newcommand{\RX}{\tx{RX}}
\renewcommand{\mod}{\tx{mod}}
\newcommand{\m}[1]{\mathbf{#1}}
\newcommand{\td}[1]{\tilde{#1}}
\newcommand{\sbf}[1]{\scriptsize{\textbf{#1}}}
\newcommand{\stxt}[1]{\scriptsize{\textrm{#1}}}
\newcommand{\suml}[2]{\sum\limits_{#1}^{#2}}
\newcommand{\sumlk}{\sum\limits_{k=0}^{K-1}}
\newcommand{\eqhsp}{\hspace{10 pt}}
\newcommand{\tx}[1]{\texttt{#1}}
\newcommand{\Hz}{\ \tx{Hz}}
\newcommand{\sinc}{\tx{sinc}}
\newcommand{\tr}{\mathrm{tr}}
\newcommand{\diag}{\mathrm{diag}}
\newcommand{\MAI}{\tx{MAI}}
\newcommand{\ISI}{\tx{ISI}}
\newcommand{\IBI}{\tx{IBI}}
\newcommand{\CN}{\tx{CN}}
\newcommand{\CP}{\tx{CP}}
\newcommand{\ZP}{\tx{ZP}}
\newcommand{\ZF}{\tx{ZF}}
\newcommand{\SP}{\tx{SP}}
\newcommand{\MMSE}{\tx{MMSE}}
\newcommand{\MINF}{\tx{MINF}}
\newcommand{\RC}{\tx{MP}}
\newcommand{\MBER}{\tx{MBER}}
\newcommand{\MSNR}{\tx{MSNR}}
\newcommand{\MCAP}{\tx{MCAP}}
\newcommand{\vol}{\tx{vol}}
\newcommand{\ah}{\hat{g}}
\newcommand{\tg}{\tilde{g}}
\newcommand{\teta}{\tilde{\eta}}
\newcommand{\heta}{\hat{\eta}}
\newcommand{\uh}{\m{\hat{s}}}
\newcommand{\eh}{\m{\hat{\eta}}}
\newcommand{\hv}{\m{h}}
\newcommand{\hh}{\m{\hat{h}}}
\newcommand{\Po}{P_{\mathrm{out}}}
\newcommand{\Poh}{\hat{P}_{\mathrm{out}}}
\newcommand{\Ph}{\hat{\gamma}}
\newcommand{\mat}[1]{\begin{matrix}#1\end{matrix}}
\newcommand{\ud}{^{\dagger}}
\newcommand{\C}{\mathcal{C}}
\newcommand{\nn}{\nonumber}
\newcommand{\nInf}{U\rightarrow \infty}

%% file: Wireless_SGLD.bbl
\begin{thebibliography}{10}

\bibitem{park2019wireless}
J.~Park, S.~Samarakoon, M.~Bennis, and M.~Debbah, ``Wireless network
  intelligence at the edge,'' {\em Proc. IEEE}, vol.~107, no.~11,
  pp.~2204--2239, 2019.

\bibitem{zhu2020toward}
G.~Zhu, D.~Liu, Y.~Du, C.~You, J.~Zhang, and K.~Huang, ``Toward an intelligent
  edge: wireless communication meets machine learning,'' {\em IEEE Commun.
  Mag.}, vol.~58, pp.~19--25, Jan. 2020.

\bibitem{sery2021over}
T.~Sery, N.~Shlezinger, K.~Cohen, and Y.~C. Eldar, ``Over-the-air federated
  learning from heterogeneous data,'' {\em IEEE Trans. Signal Process.},
  vol.~69, pp.~3796--3811, June 2021.

\bibitem{yang2021revisiting}
H.~H. Yang, Z.~Chen, T.~Q. Quek, and H.~V. Poor, ``Revisiting analog
  over-the-air machine learning: The blessing and curse of interference,'' {\em
  [Online]. Available: https://arxiv.org/pdf/2107.11733.pdf}, 2021.

\bibitem{zhu2019broadband}
G.~Zhu, Y.~Wang, and K.~Huang, ``Broadband analog aggregation for low-latency
  federated edge learning,'' {\em IEEE Trans. Wireless Commun.}, vol.~19,
  pp.~491--506, Oct. 2019.

\bibitem{liu2020privacy}
D.~Liu and O.~Simeone, ``Privacy for free: Wireless federated learning via
  uncoded transmission with adaptive power control,'' {\em IEEE J. Sel. Areas
  Commun.}, vol.~39, pp.~170--185, Nov. 2020.

\bibitem{zhu2020one}
G.~Zhu, Y.~Du, D.~G{\"u}nd{\"u}z, and K.~Huang, ``One-bit over-the-air
  aggregation for communication-efficient federated edge learning: Design and
  convergence analysis,'' {\em IEEE Trans. Wireless Commun.}, vol.~20,
  pp.~2120--2135, March 2021.

\bibitem{dwork2014algorithmic}
C.~Dwork, A.~Roth, {\em et~al.}, ``The algorithmic foundations of differential
  privacy,'' {\em Foundations and Trends{\textregistered} in Theoretical
  Computer Science}, vol.~9, no.~3--4, pp.~211--407, 2014.

\bibitem{koda2020differentially}
Y.~Koda, K.~Yamamoto, T.~Nishio, and M.~Morikura, ``Differentially private
  aircomp federated learning with power adaptation harnessing receiver noise,''
  in {\em Proc. IEEE Glob. Commun. Conf. (GLOBECOM)}, (Virtual), Dec. 2020.

\bibitem{guo2017calibration}
C.~Guo, G.~Pleiss, Y.~Sun, and K.~Q. Weinberger, ``On calibration of modern
  neural networks,'' in {\em Proc. Intl. Conf. Mach. Learning (ICML)}, (Sydney,
  Australia), pp.~1321--1330, Aug. 2017.

\bibitem{lakshminarayanan2016simple}
B.~Lakshminarayanan, A.~Pritzel, and C.~Blundell, ``Simple and scalable
  predictive uncertainty estimation using deep ensembles,'' {\em [Online].
  Available: https://arxiv.org/pdf/1612.01474.pdf}, 2016.

\bibitem{angelino2016patterns}
E.~Angelino, M.~J. Johnson, and R.~P. Adams, ``Patterns of scalable bayesian
  inference,'' {\em [Online]. Available: https://arxiv.org/pdf/1602.05221.pdf},
  2016.

\bibitem{ma2015complete}
Y.-A. Ma, T.~Chen, and E.~Fox, ``A complete recipe for stochastic gradient
  {MCMC},'' in {\em Proc. Adv. Neural Info. Proc. Syst. (NIPS)}, vol.~28,
  (Montreal, Canada), Dec. 2015.

\bibitem{dalalyan2020sampling}
A.~S. Dalalyan and L.~Riou-Durand, ``On sampling from a log-concave density
  using kinetic langevin diffusions,'' {\em Bernoulli}, vol.~26, no.~3,
  pp.~1956--1988, 2020.

\bibitem{zou2021convergence}
D.~Zou and Q.~Gu, ``On the convergence of hamiltonian monte carlo with
  stochastic gradients,'' in {\em Proc. Conf. Mach. Learning (ICML)},
  (Virtual), pp.~13012--13022, July 2021.

\bibitem{alistarh2017qsgd}
D.~Alistarh, D.~Grubic, J.~Li, R.~Tomioka, and M.~Vojnovic, ``{QSGD}:
  Communication-efficient {SGD} via gradient quantization and encoding,'' in
  {\em Proc. Adv. Neural Info. Proc. Syst. (NIPS)}, (Long Beach, USA), Dec.
  2017.

\bibitem{sery2020analog}
T.~Sery and K.~Cohen, ``On analog gradient descent learning over multiple
  access fading channels,'' {\em IEEE Trans. Signal Process.}, vol.~68,
  pp.~2897--2911, 2020.

\bibitem{fan2021temporal}
D.~Fan, X.~Yuan, and Y.-J.~A. Zhang, ``Temporal-structure-assisted gradient
  aggregation for over-the-air federated edge learning,'' {\em [Online].
  Available: https://arxiv.org/pdf/2103.02270.pdf}, 2021.

\bibitem{lee2020bayesian}
S.~Lee, C.~Park, S.-N. Hong, Y.~C. Eldar, and N.~Lee, ``Bayesian federated
  learning over wireless networks,'' {\em [Online]. Available:
  https://arxiv.org/pdf/2012.15486.pdf}, 2020.

\bibitem{chen2020fedbe}
H.-Y. Chen and W.-L. Chao, ``Fedbe: Making bayesian model ensemble applicable
  to federated learning,'' in {\em Proc. Intl. Conf. Learning Representations
  (ICLR)}, (Virtual), May 2021.

\bibitem{zhu2021data}
Z.~Zhu, J.~Hong, and J.~Zhou, ``Data-free knowledge distillation for
  heterogeneous federated learning,'' {\em [Online]. Available:
  https://arxiv.org/pdf/2105.10056.pdf}, 2021.

\bibitem{corinzia2019variational}
L.~Corinzia and J.~M. Buhmann, ``Variational federated multi-task learning,''
  {\em [Online]. Available: https://arxiv.org/pdf/1906.06268.pdf}, 2019.

\bibitem{kassab2020federated}
R.~Kassab and O.~Simeone, ``Federated generalized bayesian learning via
  distributed stein variational gradient descent,'' {\em [Online]. Available:
  https://arxiv.org/pdf/2009.06419.pdf}, 2020.

\bibitem{elfederated}
K.~el~Mekkaoui, D.~Mesquita, P.~Blomstedt, and S.~Kaski, ``Federated stochastic
  gradient langevin dynamics,'' in {\em Proc. 37th Conf. Uncertain. Artif.
  Intell. (UAI)}, (Virtual), July 2021.

\bibitem{vono2021qlsd}
M.~Vono, V.~Plassier, A.~Durmus, A.~Dieuleveut, and E.~Moulines, ``{QLSD}:
  {Quantised} {Langevin} stochastic dynamics for {Bayesian} federated
  learning,'' {\em [Online]. Available: https://arxiv.org/pdf/2106.00797.pdf},
  2021.

\bibitem{gandikota2019vqsgd}
V.~Gandikota, R.~K. Maity, and A.~Mazumdar, ``{vqSGD}: Vector quantized
  stochastic gradient descent,'' {\em [Online]. Available:
  https://arxiv.org/pdf/1911.07971.pdf}, 2019.

\bibitem{agarwal2018cpsgd}
N.~Agarwal, A.~T. Suresh, F.~X.~X. Yu, S.~Kumar, and B.~McMahan, ``cp{SGD}:
  Communication-efficient and differentially-private distributed {SGD},'' in
  {\em Proc. Adv. Neural Info. Proc. Syst. (NIPS)}, (Montreal, Canada), Dec.
  2018.

\bibitem{seif2020wireless}
M.~Seif, R.~Tandon, and M.~Li, ``Wireless federated learning with local
  differential privacy,'' in {\em Proc. IEEE Intl. Symp. Info. Theory (ISIT)},
  (Los Angeles, USA), June 2020.

\bibitem{zhang2021turning}
Z.~Zhang, G.~Zhu, R.~Wang, V.~K. Lau, and K.~Huang, ``Turning channel noise
  into an accelerator for over-the-air principal component analysis,'' {\em
  [Online]. Available: https://arxiv.org/pdf/2104.10095.pdf}, 2021.

\bibitem{wang2015privacy}
Y.-X. Wang, S.~Fienberg, and A.~Smola, ``Privacy for free: Posterior sampling
  and stochastic gradient {Monte} {C}arlo,'' in {\em Proc. Conf. Mach. Learning
  (ICML)}, (Lille, France), July 2015.

\bibitem{li2019connecting}
B.~Li, C.~Chen, H.~Liu, and L.~Carin, ``On connecting stochastic gradient
  {MCMC} and differential privacy,'' in {\em Proc. Intl. Conf. Artif. Intell.
  Stat. (AISTATS)}, (Naha, Japan), April 2019.

\bibitem{dalalyan2017further}
A.~Dalalyan, ``Further and stronger analogy between sampling and optimization:
  Langevin monte carlo and gradient descent,'' in {\em Proc. Conf. Learning
  Theory (COLT)}, (Amsterdam, Netherlands), July 2017.

\bibitem{liu2021channel}
D.~Liu and O.~Simeone, ``Channel-driven {M}onte {C}arlo sampling for {B}ayesian
  distributed learning in wireless data centers,'' {\em [Online]. Available:
  https://arxiv.org/pdf/2103.01351.pdf}, 2021.

\bibitem{rabinovich2015variational}
M.~Rabinovich, E.~Angelino, and M.~I. Jordan, ``Variational consensus {M}onte
  {C}arlo,'' in {\em Proc. Adv. Neural Info. Proc. Syst. (NIPS)}, (Montreal,
  Canada), Dec. 2015.

\bibitem{bhattacharya1978criteria}
R.~Bhattacharya {\em et~al.}, ``Criteria for recurrence and existence of
  invariant measures for multidimensional diffusions,'' {\em The Annals of
  Probability}, vol.~6, no.~4, pp.~541--553, 1978.

\bibitem{chatterji2018theory}
N.~Chatterji, N.~Flammarion, Y.~Ma, P.~Bartlett, and M.~Jordan, ``On the theory
  of variance reduction for stochastic gradient monte carlo,'' in {\em Proc.
  Intl. Conf. Mach. Learning (ICML)}, (Stockholm, Sweden), pp.~764--773, July
  2018.

\bibitem{mahmood2019time}
A.~Mahmood, M.~I. Ashraf, M.~Gidlund, J.~Torsner, and J.~Sachs, ``Time
  synchronization in 5{G} wireless edge: Requirements and solutions for
  critical-mtc,'' {\em IEEE Commun. Mag.}, vol.~57, pp.~45--51, Dec. 2019.

\bibitem{arjovsky2017wasserstein}
M.~Arjovsky, S.~Chintala, and L.~Bottou, ``Wasserstein generative adversarial
  networks,'' in {\em Proc. Intl. Conf. Mach. Learning (ICML)}, (Sydney,
  Australia), pp.~214--223, Aug. 2017.

\bibitem{bubeck2014convex}
S.~Bubeck, ``Convex optimization: Algorithms and complexity,'' {\em Found.
  Trends Mach. Learn.}, vol.~8, pp.~231--357, Nov. 2015.

\bibitem{chen2020understanding}
X.~Chen, Z.~S. Wu, and M.~Hong, ``Understanding gradient clipping in private
  {SGD}: A geometric perspective,'' in {\em Proc. Adv. Neural Info. Proc. Syst.
  (NIPS)}, (Virtual), Dec. 2020.

\bibitem{ren2020scheduling}
J.~Ren, Y.~He, D.~Wen, G.~Yu, K.~Huang, and D.~Guo, ``Scheduling for cellular
  federated edge learning with importance and channel awareness,'' {\em IEEE
  Trans. Wireless Commun.}, vol.~19, no.~11, pp.~7690--7703, 2020.

\bibitem{liu2020data}
D.~Liu, G.~Zhu, J.~Zhang, and K.~Huang, ``Data-importance aware user scheduling
  for communication-efficient edge machine learning,'' {\em IEEE Transactions
  on Cognitive Communications and Networking}, vol.~7, no.~1, pp.~265--278,
  2020.

\bibitem{WFLMCarxiv}
D.~Liu and O.~Simeone, ``Wireless federated langevin monte carlo: Repurposing
  channel noise for bayesian sampling and privacy,'' {\em [Online]. Available:
  https://arxiv.org/pdf/2108.07644.pdf}, 2021.

\bibitem{sun2018gaussian}
C.~Sun, H.~Yan, X.~Qiu, and X.~Huang, ``Gaussian word embedding with a
  wasserstein distance loss,'' {\em [Online]. Available:
  https://arxiv.org/pdf/1808.07016.pdf}, 2018.

\bibitem{hein2005intrinsic}
M.~Hein and J.-Y. Audibert, ``Intrinsic dimensionality estimation of
  submanifolds in rd,'' in {\em Proc. Intl. Conf. Mach. Learning (ICML)},
  (Bonn, Germany), Aug. 2005.

\bibitem{wei2021federated}
X.~Wei and C.~Shen, ``Federated learning over noisy channels: Convergence
  analysis and design examples,'' {\em [Online]. Available:
  https://arxiv.org/pdf/2101.02198.pdf}, 2021.

\bibitem{amiri2021convergence}
M.~M. Amiri, D.~G{\"u}nd{\"u}z, S.~R. Kulkarni, and H.~V. Poor, ``Convergence
  of federated learning over a noisy downlink,'' {\em IEEE Transactions on
  Wireless Communications}, 2021.

\end{thebibliography}
